\documentclass[10pt]{article}
\usepackage[margin=1in]{geometry}

\usepackage[T1]{fontenc}

\usepackage{amsmath, amsfonts, amssymb}
\usepackage{mathtools}
\usepackage{amsthm}

\usepackage{hyperref}        
\usepackage[nameinlink, noabbrev, capitalise]{cleveref} 

\usepackage{aliascnt}                          

\newtheorem{theorem}{Theorem}
\crefname{theorem}{Theorem}{Theorems}

\newcommand{\newtheoremalias}[3]{%
  \newaliascnt{#1}{theorem}
  \newtheorem{#1}[#1]{#2}
  \aliascntresetthe{#1}
  \crefname{#1}{#2}{#3}
}

\newtheoremalias{lemma}{Lemma}{Lemmas}
\newtheoremalias{claim}{Claim}{Claims}
\newtheoremalias{conjecture}{Conjecture}{Conjectures}
\newtheoremalias{corollary}{Corollary}{Corollaries}
\newtheoremalias{proposition}{Proposition}{Propositions}
\newtheoremalias{definition}{Definition}{Definitions}
\newtheoremalias{example}{Example}{Examples}
\newtheoremalias{observation}{Observation}{Observations}

\usepackage{graphicx}
\usepackage{pgfplots}
\pgfplotsset{compat=1.16} 

\usepackage{tikz}
\usepackage{tikz-3dplot}
\pgfdeclarelayer{background}
\pgfsetlayers{background,main}

\usepackage{booktabs}
\usepackage{multirow}
\usepackage{enumitem}
\usepackage{caption} 
\usepackage{subcaption}
\usepackage[ruled,vlined,linesnumbered]{algorithm2e}
\crefname{algocf}{Algorithm}{Algorithms}

\usepackage{xcolor} 

\SetCommentSty{mycommfont}

\usepackage{xspace}
\usepackage{relsize}
\usepackage{verbatim}
\usepackage{listings}
\usepackage{setspace} 
\usepackage{siunitx}  
\usepackage{colortbl} 

\date{}
\title{Algorithms for Optimizing Acyclic Queries}
\author{
  Zheng Luo \\ UCLA \\ {\small\texttt{luo@cs.ucla.edu}}
  \and
  Wim Van den Broeck \\ University of Bergen \\ {\small\texttt{wim.broeck@uib.no}}
  \and
  Guy Van den Broeck \\ UCLA \\ {\small\texttt{guyvdb@cs.ucla.edu}}
  \and
  Yisu Remy Wang \\ UCLA \\ {\small\texttt{remywang@cs.ucla.edu}}
}

\newcommand{\LCA}{\mathsf{LCA}}
\newcommand{\LA}{\mathsf{LA}}
\newcommand{\LCAE}{\lambda}

\newcommand{\defeq}{:=}  

\newcommand{\parent}{p}
\newcommand{\children}{c}
\newcommand{\siblings}{s}

\newcommand{\setof}[2]{\left\{#1 \mid #2\right\}}

\DeclareMathOperator*{\argmax}{arg\,max}

\newcommand{\buildEG}{buildEG}

\newcommand{\cut}[1]{}

\newcommand{\hyp}{H} 
\newcommand{\lin}{L} 
\newcommand{\EG}{G^\equiv} 
\newcommand{\vars}{\chi} 
\newcommand{\nodes}{\rho} 
\newcommand{\nodesEG}{\rho^\equiv} 
\newcommand{\slide}[1]{\leq_{#1}} 
\newcommand{\w}{\omega} 
\newcommand{\T}{T} 
\newcommand{\G}{G} 
\newcommand{\JTof}[1]{\mathcal{T}(#1)}
\newcommand{\depth}{\mathsf{d}}
\newcommand{\dist}{\mathsf{dist}} 
\newcommand{\MWJT}{\hat{T}}  
\newcommand{\orphan}[1]{\mathring{#1}}
\newcommand{\prefix}[1]{\mathsf{pre}({#1})}
\newcommand{\x}{x}   
\newcommand{\X}{X}   
\newcommand{\rel}{r} 
\newcommand{\R}{R}   
\newcommand{\e}{e}   
\newcommand{\E}{E}   
\newcommand{\Power}{\mathcal{P}} 
\newcommand{\bigO}{\mathcal{O}}
\newcommand{\func}{\mathsf{f}} 

\newcommand{\spath}{\mathsf{P}}   
\newcommand{\id}{\texttt{id}}   

\newcommand{\review}[1]{} 
\newcommand{\revA}[1]{{#1}} 
\newcommand{\revB}[1]{{#1}} 
\newcommand{\revC}[1]{{#1}} 

\begin{document}

\maketitle

\begin{abstract}
Most research on query optimization has centered on binary join algorithms
like hash join and sort-merge join. However, recent years have seen growing
interest in theoretically optimal algorithms, notably Yannakakis' algorithm.
These algorithms rely on {\em join trees}, which differ from the operator
trees for binary joins and require new optimization techniques. We
propose three approaches to constructing join trees for acyclic queries.
First, we give an algorithm to enumerate all join trees of an
$\alpha$-acyclic query {\em by edits}
in linear time with amortized constant delay, which
forms the basis of a cost-based optimizer for acyclic joins.
Second, we show the Maximum Cardinality Search algorithm by Tarjan and Yannakakis 
 constructs the unique {\em shallowest} join tree
 for any Berge-acyclic query,
 thus enabling parallel execution of large join queries.
Finally, we prove that a simple algorithm by Hu et al.\ 
converts any connected left-deep linear plan of a $\gamma$-acyclic 
query into a join tree, allowing reuse of
optimizers developed for binary joins.
\end{abstract}

\section{Introduction}

The query optimizer sits at the heart of a database system.
It takes a query as input and generates a plan for efficient execution,
allowing the user to program declaratively without worrying about performance.
Among the many relational algebra operators,
 join has received significant attention in optimization research.
Its {\em compositional} nature allows for 
 combining information from multiple relations,
 constructing complex queries from simple ones,
 and producing an output asymptotically larger than the inputs.
The primary challenge is the {\em join ordering problem}
 to find the best arrangement of many join operations.
Most existing research has focused on binary
join algorithms such as hash join and sort-merge join, but these can
produce unnecessarily large intermediates.
Recent work has revived
interest in optimal join algorithms, notably Yannakakis’
instance-optimal algorithm~\cite{10.5555/1286831.1286840}
for acyclic queries, which runs in linear time in the input and output size,
$\bigO(|\textsf{IN}| + |\textsf{OUT}|)$. 
Its execution is guided by
{\em join trees} whose nodes are relations,
different from traditional binary join plans with 
relations at the leaves and join operators at the internal nodes.
Although the algorithm is optimal regardless of the join
tree, the choice of plan can affect practical performance.
In this paper, we study the optimization problem in the context of Yannakakis-style algorithms.

A query optimizer typically has two parts: a {\em plan generator} and 
a {\em cost model} to assess each plan. This paper
focuses on plan generation and presents three approaches:
\begin{enumerate}
    \item \revB{Given the line graph $\lin$ (\cref{def:line-graph}) of an $\alpha$-acyclic query $\hyp$,
    \revB{\cref{algo:buildEG}} 
    enumerates all join trees {\em by edits}\footnote{
    To avoid redundant work, 
    enumeration {\em by edits}~\cite{Kapoor1995} 
    outputs the difference between consecutive elements.}
    with amortized constant delay, 
    i.e., in $\bigO(|\lin| + k)$ where $k$ is the number of trees 
    generated \revB{(\cref{thm:alpha_JTE_enum})}.
    If the query is $\gamma$-acyclic, 
    \revB{\cref{algo:buildEG4gamma}} further reduces
    the total time complexity
    to $\bigO(|\hyp| + k)$ \revB{(\cref{thm:gamma_JTE_enum})}.}
    \item \revB{Given a Berge-acyclic query,
    we prove that the classic Maximum Cardinality Search algorithm
    by Tarjan and Yannakakis~\cite{tarjan1984simple}
    yields a unique {\em shallowest} join tree \revB{(\cref{thm:canonical_join_tree})}, 
    enabling parallel execution of very large queries.}
    \item \revB{Given a connected left-deep linear join plan
    for any $\gamma$-acyclic query,
    we prove that a simple algorithm by Hu et al.~\cite{hu2024treetracker}
    always converts the plan into a valid join tree \revB{(\cref{cor:bin2tree})},
    allowing reuse of existing optimizers.  }
\end{enumerate}




The rest of the paper is organized as follows:
\cref{sec:related} discusses related work;
\cref{section:preliminaries} introduces relevant concepts and notations;
\cref{section:enum} presents the join tree enumeration algorithms;
\cref{section:cjt} introduces the unique shallowest join tree, namely the canonical join tree and its construction;
\cref{section:frombinplan} discusses the conversion from binary plans to join trees;
\cref{section:conclusion} concludes the paper by pointing to avenues for future work.
Additional technical details, proofs and a notation table (\cref{tab:notation-prelims-enum}) 
can be found in the appendix.
\section{Related Work}\label{sec:related}
Join order optimization is well studied, with algorithms based on
dynamic programming (DP) from the bottom up
\cite{10.5555/1182635.1164207,ref24_selinger1979access,ref23_neumann2009query},
cost-based rewriting from the top down
\cite{ref6_dehaan2007optimal,ref9_fender2013counter}, greedy heuristics
\cite{ref2_bruno2010polynomial,ref8_fegaras1998heuristic,ref27_swami1989optimization},
and randomized search \cite{ref26_steinbrunn1997heuristic}. Since the
plan space is exponential, most methods prune it: some restrict to
left-deep plans~\cite{DBLP:journals/tods/IbarakiK84,DBLP:conf/vldb/KrishnamurthyBZ86},
while others avoid Cartesian products
\cite{10.5555/1182635.1164207,DBLP:conf/sigmod/MoerkotteN08}.
Our algorithms restrict the query plans
 to those running in linear time for acyclic queries.
In particular, while avoiding Cartesian products requires each subplan
 to form a {\em spanning tree} of the corresponding subquery's join graph,%
\footnote{The join graph of a query has a vertex for each relation
 and an edge for each pair of relations that join with each other.
 We later define this as the line graph of the query hypergraph in \cref{def:line-graph}.}
 our algorithms find {\em maximum spanning trees} of the {\em weighted} join graph.
Several algorithms for ordering binary joins are based on dynamic programming
 and tabulate shared structures among different plans~\cite{
    10.5555/1182635.1164207,DBLP:conf/sigmod/MoerkotteN08}.
This is desirable because subplans are grouped into equivalence classes,
 and the optimal plan can be constructed in a bottom-up manner.
Future work may explore constructing compact representations of join trees,
 suitable for dynamic programming.

Several recent papers have proposed practical implementations of
 Yannakakis' algorithm for acyclic queries~\cite{DBLP:journals/pvldb/BekkersNVW25,hu2024treetracker,DBLP:journals/pacmmod/WangCDYLL25,zhao2025debunking}.
For example, Zhao et al.~\cite{zhao2025debunking} find that
 different query plans perform similarly,
 thanks to the optimality of Yannakakis' algorithm.
They adopt a simple heuristic to construct the join tree
 by picking the largest input relation as the root,
 and then greedily attaching the remaining relations into the tree.
Inspired by this algorithm, we prove in~\cref{section:cjt} that for Berge-acyclic queries
 there is a unique {\em shallowest} join tree for any given root where 
 the depth of each node is minimized.
Furthermore, this tree can be constructed in linear time by 
Tarjan and Yannakakis'
 Maximum Cardinality Search algorithm~\cite{tarjan1984simple}.
Shallow trees are desirable for parallel execution, 
 where the depth of the tree determines the number of sequential steps.
Other practical implementations of Yannakakis' algorithm
 leverage existing optimizers for binary joins
 and convert a binary plan into a join tree~\cite{DBLP:journals/pvldb/BekkersNVW25,hu2024treetracker}.
In particular, Hu et al.~\cite{hu2024treetracker} find that every left-deep linear plan
 encountered in practice can be converted into a join tree
 by a simple algorithm.
This is not surprising, as we will prove in~\cref{section:frombinplan} that every connected left-deep linear plan
 of a $\gamma$-acyclic query must traverse some join tree from root to leaves.

On the theoretical side, attention has been focused on finding
 {\em (hyper-)tree decompositions} to improve
 the asymptotic complexity of query processing~\cite{DBLP:journals/pacmmod/SurianarayananMSL25,DBLP:journals/pvldb/HeY24,DBLP:journals/tods/GottlobLOP24}.
The general goal is to find a decomposition with small {\em width}
 which can be used to guide the execution of join algorithms.
Most algorithms find a single decomposition with minimum width
 to achieve the optimal asymptotic complexity~\cite{DBLP:journals/pacmmod/SurianarayananMSL25,DBLP:journals/pvldb/HeY24,DBLP:journals/tods/GottlobLOP24}.
Nevertheless, different decompositions with the same width
 may still lead to different performance in practice,
 and cost-based optimization remains crucial.
For this, Carmeli et al.~\cite{DBLP:conf/pods/CarmeliKK17} propose an algorithm to enumerate
 tree decompositions with polynomial delay.
In this paper, we focus on acyclic queries and their join trees,
 which are precisely the decompositions with width 1.
Our enumeration algorithm can generate all join trees by edits
 with amortized constant delay.

\section{Preliminaries}\label{section:preliminaries}

We focus on full conjunctive queries~\cite{DBLP:books/aw/AbiteboulHV95} in this paper
 and identify each query with its hypergraph,
 where each vertex represents a variable
 and each hyperedge represents a relation.
\begin{definition}[Hypergraph]
A {\em hypergraph} $\hyp = (\X, \R, \vars)$ consists of
 a set of vertices $\X$, a set of hyperedges $\R$ and an incidence function
 $\vars: \R \to \Power(\X)$.
\end{definition}
We only consider hypergraphs without isolated vertices,
empty hyperedges or duplicated hyperedges over the same set of vertices.
We assume each hyperedge $\rel$ contains a bounded number ($\bigO(1)$)
 of vertices exclusive to $\rel$.
When there is no ambiguity, we will use $\rel$ interchangeably with $\vars(\rel)$,
$\x \in \rel$ with $\x \in \vars(\rel)$, and $\hyp(\X, \R)$ with $\hyp(\X, \R, \vars)$.
We will also apply common set operations directly to hyperedges,
  e.g., $\rel_1 \cap \rel_2$ for $\vars(\rel_1) \cap \vars(\rel_2)$.
\revB{
In fact, the reader can often ignore $\vars$ and identify
 a hyperedge with  its set of vertices;
 the purpose of $\vars$ is to relate hyperedges over different
 vertex sets across hypergraphs in the correctness proofs of
 our algorithms. }
%
%
\revB{
We write $\X(\hyp)$ and $\R(\hyp)$ to refer to the vertices 
 and edges of $\hyp$ respectively.}
  The size of a hypergraph is the total size of all hyperedges 
 $|\hyp| = \sum_{\rel \in \R}|\rel|$. 
 %
 We let $\hyp|_{\x}$ denote the {\em neighborhood} of $\x \in \X$, 
 consisting of all hyperedges containing $\x$. 
 

%
%

 %
%
\begin{definition}[Multigraph]
    A multigraph $\G = (\R, \E, \nodes)$
    consists of a set of vertices $\R$,
     a set of edges $\E$, 
    and the incidence function $\nodes: \E \to \Power(\R)$ 
    such that $1 \leq |\nodes(\e)|\leq 2$ for any $\e \in \E$.
    An edge $\e$ is a {\em self-loop} if $|\nodes(\e)|=1$.
    Edges $\e_1, \e_2 \in \E$ are {\em parallel} 
    if $\nodes(\e_1) = \nodes(\e_2)$.
\end{definition}
We denote the set of vertices in $\G$ with $\R$,
 because we will soon define the {\em line graph} $\lin$
 of a hypergraph where each vertex in $\lin$ represents a hyperedge.
We write $\R(\G)$ and $\E(\G)$ to refer to the vertices 
 and edges of $\G$ respectively.

\begin{definition}[Simple Graph]\label{def:simple_graph}
    A simple graph is a multigraph with 
    no parallel edges or self-loops,
    i.e., the incidence function is injective and always returns two distinct vertices.
\end{definition}
Because each edge $\e$ in a simple graph can be identified with
 its two endpoints $\nodes(\e) = \{\rel_1, \rel_2\}$,
 we will omit $\nodes$ and write $\{\rel_1, \rel_2\}$ for $\e$ when there is no ambiguity.

\begin{definition}[Cycle, Clique and Diamond]\label{def:diamond}
\revB{
In a simple graph $\G = (\R, \E)$, a sequence of $n \geq 3$ distinct vertices
$(v_0, \ldots, v_{n-1})$ where $v_i \in \R$ is}

\begin{itemize}
    \item a {\em cycle} if there is an edge between 
    $v_i$ and $v_{(i+1) \bmod n}$ 
    for all $0 \leq i \leq n-1$;
    \item an {\em n-clique $K_n$} if the induced graph 
    $G|_{\{v_0, \dots, v_n\}}$ has an edge between each pair of vertices;
    \item a {\em diamond} if $n = 4$ and
    $G|_{\{v_0, \dots, v_3\}}$ has one fewer edge than a $K_4$.
\end{itemize}
\end{definition}

A {\em weighted graph} is a (multi-)graph where each edge is assigned a weight:
\begin{definition}[Weighted Graph]
    A weighted graph $(\G, \w)$ consists of a multigraph $\G = (\R, \E, \nodes)$ and weight function
    $\w: \E \to \mathbb{N}$ assigning a natural number to each edge in $\G$.
\end{definition}

The {\em line graph} of a hypergraph is the ``intersection graph'' of its hyperedges:
\begin{definition}[Line Graph]\label{def:line-graph}
For a hypergraph $\hyp$,
 the {\em line graph} $\lin(\hyp) = (\G, \vars)$
 consists of a simple graph $\G = (\R, \E)$
 and an {\em edge labeling function} $\vars: \E \to \Power(\X)$.
The vertices of $\G$ are the hyperedges $\R$ of $\hyp$,
 and there is an edge $\e = \{\rel_1, \rel_2\} \in \E$
 when $\rel_1 \cap \rel_2 \neq \emptyset$.
The edge labeling function maps each edge $\e$ to the intersection
 of the hyperedges represented by the endpoints of $\e$:
 $\vars(\{\rel_1, \rel_2\}) = \rel_1 \cap \rel_2$.
\end{definition}
In addition, we define a weight function mapping each edge to the size of its label:
\begin{definition}[Line Graph Edge Weight]
For each edge $\e \in \E(\lin(\hyp))$,
 the weight function $\w: \E(\lin(\hyp)) \to \mathbb{N}$ is defined as
 $\w(\e) = |\vars(\e)|$.
\end{definition}
When there is no ambiguity from the context, we simply write $\lin(\hyp)$ as $\lin$,
 and we may treat the line graph $(\G, \vars)$ as a weighted graph $(\G, \w)$.
The size of a line graph is the sum of all edge weights $|\lin| = \sum_{\e \in \E(\lin)} \w(\e)$. 
For any subgraph $\G$ of $\lin(\hyp)$,
 we let $\G|_\x$ denote the subgraph of $\G$ induced by $\hyp|_\x$.

\begin{example}
\cref{subfig:hyp_lin_jt_h6} shows an example hypergraph of size 14.
The line graph $\lin_6 = \lin(\hyp_6)$ is shown in \cref{subfig:hyp_lin_jt_l6}.
For instance, hyperedges $Y$ and $U$ share two vertices $a$ and $d$, 
so they are connected by an edge in the line graph.
The edge weight is given by $\w(\{Y, U\}) = |\{a, d\}| = 2$.
Each unannotated edge in $\lin_6$ connects two hyperedges sharing only the vertex $a$.
Otherwise, the common vertices shared by a pair of hyperedges are 
annotated next to the corresponding edge.
The size of the line graph is $|\lin_6| = 19$.
\end{example}

\begin{definition}[Join Tree]\label{def:join_tree}
    A {\em join tree} $\T$ of hypergraph $\hyp$ is a spanning tree of $\lin(\hyp)$
    such that $\T|_\x$ is a connected subtree for each $\x \in \X(\hyp)$. 
If a certain vertex is specified as the {\em root},
 $\T$ becomes a {\em rooted join tree}.
\end{definition}
We use $\R(\T)$ to denote the set of nodes and $\E(\T)$ to
denote the set of edges of tree $\T$. 
The requirement that $\T|_{\x \in \X(\hyp)}$ be a connected subtree
is also known as the {\em running
intersection property}~\cite{fagin1983degrees}.
%
We write $\JTof{\lin(\hyp)}$ or $\JTof{\hyp}$ 
 to denote the set of unrooted join trees of $\hyp$.
%
%
We use $\T_{\rel}$ to denote a join tree rooted at $\rel \in \R(\T)$. 
When there is no ambiguity, we simply write $\T$.
The depth of a node $\rel_i$ in the rooted tree $\T_{\rel}$,
 denoted $\depth(\T_{\rel}, \rel_i)$,
 is defined as its distance from the root. 

\revA{
Four common notions of hypergraph acyclicity are defined
 in increasing order of strictness~\cite{fagin1983degrees}, namely
{\em $\alpha$-acyclic $\supset$ $\beta$-acyclic $\supset$
 $\gamma$-acyclic $\supset$ Berge-acyclic.} }

\begin{definition}[Hypergraph Acyclicity]\label{def:acyclicity}
A hypergraph $\hyp$ is:
\begin{itemize}
    \item {\em $\alpha$-acyclic} if it admits a join tree as in \cref{def:join_tree};
    \item {\em $\beta$-acyclic} if every subgraph of $\hyp$ is $\alpha$-acyclic;
    \item {\em $\gamma$-acyclic} if it does not contain any {\em $\gamma$ cycle}.
        A $\gamma$ cycle is a sequence of length $k \geq 3$ of distinct hyperedges
        and distinct vertices $(\rel_0, \x_0, \dots, \rel_{k-1}, \x_{k-1})$ such that 
        every $\x_{i\in [0, k-2]}$ belongs to $\rel_i \cap \rel_{i+1}$ 
        and no other $\rel_j$ while $\x_{k-1}$ belongs to $\rel_{k-1} \cap \rel_0$ and possibly 
        other hyperedges;
    \item {\em Berge-acyclic} if it does not contain any {\em Berge cycle}.
        A Berge cycle is a sequence of length $k \geq 2$ of distinct vertices and distinct hyperedges 
        $(\rel_0, \x_0, \dots, \rel_{k-1}, \x_{k-1})$
        such that $\forall\ i \in [k]: \x_i \in \rel_i \cap \rel_{(i+1)\bmod k}$.
\end{itemize}
\end{definition}

The following classic result relates join trees of $\hyp$ to maximum spanning trees
 of $\lin(\hyp)$:
\begin{theorem}[Maier~\cite{maier1983theory}]\label{thm:join_tree_mst}
Given an $\alpha$-acyclic hypergraph $\hyp = (\X, \R)$,
 a tree with nodes in $\R$ is a join tree of $\hyp$
 if and only if it is a maximum spanning tree of $\lin(\hyp)$.
\end{theorem}

Join trees for acyclic hypergraphs can be constructed by a procedure called {\em GYO reduction}.
\begin{definition}
A {\em GYO reduction order} is a sequence of hyperedges
 $\rel_1, \rel_2, \ldots, \rel_k$ such that for each $\rel_{i < k}$,
 there is some $\rel_{p > i}$, called the {\em parent} of $\rel_i$,
 such that $\forall \rel_{j > i} : \rel_i \cap \rel_j \subseteq \rel_p$.
\end{definition}
The GYO reduction algorithm~\cite{GYO762509} finds such an order iteratively,
 and attaches each hyperedge to its parent to form a join tree.
It generates the join tree $\T^G$ as shown in \cref{subfig:hyp_lin_jt_tgyo},
where, for example, $T$ is the parent of $W$, denoted as $\parent(W) = T$ 
and the root $P$ has no parent.

\begin{figure}[tbp]
  \refstepcounter{figure} 
  \centering
  \begin{minipage}[c]{0.45\textwidth}
    \input{algos/algo_MCS.tex}
  \end{minipage}
  \hfill
  \begin{minipage}[c]{0.52\textwidth}

    \hfill
    \begin{subfigure}[b]{.45\textwidth}
    \centering
    \includegraphics[height=2.3cm,,keepaspectratio]{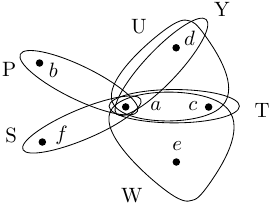}
    \caption{Hypergraph $\hyp_6$}
    \label{subfig:hyp_lin_jt_h6}
    \end{subfigure}
    \begin{subfigure}[b]{.45\textwidth}
    \centering
    \includegraphics[height=2.4cm,,keepaspectratio]{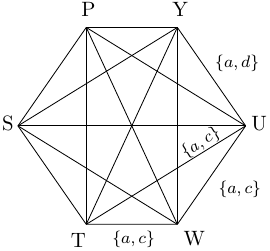}
    \caption{Line graph $\lin_6$}
    \label{subfig:hyp_lin_jt_l6}
    \end{subfigure}


    \hfill
    \begin{subfigure}[b]{.45\textwidth}
    \centering
    \includegraphics[height=2.5cm,,keepaspectratio]{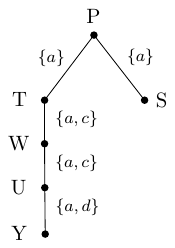}
    \caption{$\T^{G}$ by GYO}
    \label{subfig:hyp_lin_jt_tgyo}
    \end{subfigure}
    \begin{subfigure}[b]{.45\textwidth}
    \centering
    \includegraphics[height=2.5cm,,keepaspectratio]{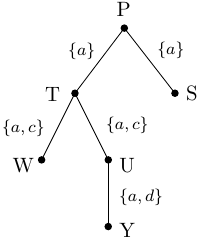}
    \caption{$\T^{M}$ by MCS}
    \label{subfig:hyp_lin_jt_tmcs}
    \end{subfigure}
    \addtocounter{figure}{-1} 
    \captionof{figure}
    {A hypergraph $\hyp_6$, its line graph $\lin_6$,
    and two join trees $\T^{G}$ and $\T^{M}$.}
  \label{fig:hyp_lin_jt}
  \end{minipage}
\end{figure}

\revB{
Another way to construct join trees 
 is via the Maximum Cardinality Search (MCS) algorithm~\cite{tarjan1984simple},
 and we present a simplified version in \cref{algo:MCS}.
The join tree produced by the algorithm is called an MCS tree.
Given an input hypergraph $\hyp$ and a hyperedge $\rel \in \R(\hyp)$
chosen as the root,
\cref{algo:MCS} constructs a rooted MCS tree 
by iteratively expanding toward hyperedges connected with
those already processed.
The algorithm maintains a set $\R$ of unprocessed hyperedges and 
a set $\X$ of unvisited vertices. 
Initially, the MCS tree contains no edges,
 and the first hyperedge to be processed is the root $\rel$.
In each iteration, the algorithm processes the current hyperedge $\rel$ by
examining all of its unvisited vertices in $\rel \cap \X$. 
Whenever such a vertex $\x$ is encountered, 
every remaining hyperedge $\rel' \in \R$ that contains $\x$
records $\rel$ as its tentative parent $\parent(\rel') \gets \rel$. 
Once all vertices of $\rel$ are marked as visited, they are removed from $\X$.
The hyperedge $\rel$ is marked as processed and removed from $\R$.
After processing each $\rel$,
the edge between $\rel$ and its recorded parent $\parent(\rel)$ 
is added to the MCS tree as $\{\rel, \parent(\rel)\}$.
The algorithm selects as the next hyperedge an
element of $\R$ that maximizes the number of already visited vertices,
namely a hyperedge $\rel$ that maximizes $|\rel \setminus \X|$, 
breaking ties arbitrarily.
This process continues until all hyperedges have been processed, when
the recorded edges $\{\rel, \parent(\rel)\}$ form an MCS tree rooted at $\rel$.}

\begin{example}\label{example:MCS}
On $\hyp_6$ in \cref{subfig:hyp_lin_jt_h6} with $r = P$,
\cref{algo:MCS} examines all vertices of $P$.
$P$ also becomes the tentative parent of five hyperedges, 
$\parent(S) = \parent(T) = \parent(U) = \parent(W) = \parent(Y) \gets P$.
\revB{It proceeds to check the next hyperedge with the most visited vertices as
in Line~\ref{algoLine:mcs_argmax}}.
At this point, each of the five has one visited vertex,
\cref{algo:MCS} breaks the tie arbitrarily, say proceeding with $S$, 
and examining all its vertices.
Its parent is finalized as $\parent(S) = P$.
Each of the remaining four again has one visited vertex, 
so \cref{algo:MCS} proceeds to check $T$.
Its parent is finalized as $\parent(T) = P$.
While checking $T$, \cref{algo:MCS} examines the vertex $c$ that is also
 in $U, W$, their parents are updated to $\parent(U) = \parent(W) \gets T$.
It continues until all hyperedges are processed.
By connecting each hyperedge with its parent, \cref{algo:MCS} constructs 
the MCS tree $\T^M$ in~\cref{subfig:hyp_lin_jt_tmcs}.
\end{example}

\cref{algo:MCS} can generate different join trees
by varying the choice of root and tie-breaking strategy.
However, it does not guarantee to generate all possible trees.
For example, it never generates the join tree $\T^G$ shown in 
\cref{subfig:hyp_lin_jt_tgyo}.

%
We extend the notions of parent, children and siblings to edges of a rooted tree. 
Let $\e \in \E(\T)$ be an edge of tree $\T$ rooted at $\rel$ (and $\e$ not incident to $\rel$).
Its parent $\parent(\e) \in \E(\T)$ is the unique tree edge
incident to $\e$ and lying closer to the root $\rel$. We let  
$\children(\e)$ denote the set of children of $\e$ in $\T$.
The siblings of $\e$ is the set of edges incident to the parent of $\e$, namely
$\siblings(\e) = \setof{\e' \in \E(\T) \setminus \{\e\}}{\parent(\e') = \parent(\e)}$.
Edges incident to the root do not have a parent and are all siblings.


\cref{lemma:MCS_rip} shows that
an MCS tree is ``somewhat monotonic''
in that every edge must contain some variable not in its parent.
If two edges share any variable not
 in their parents, they must be siblings.
\begin{lemma}\label{lemma:MCS_rip}
\revB{Let $\T$ be an MCS tree. For an edge $\e \in \T$ that has a parent, 
then}
\begin{enumerate}
    \item $\e \not \subseteq \parent(\e)$
    \item For another edge
    $
    \e' \in \T:
        (\e \setminus \parent(\e))\
        \cap\
        (\e' \setminus \parent(\e'))\
        \neq\ \emptyset 
        \quad \implies \quad 
        \parent(\e) = \parent(\e').
    $
\end{enumerate}
\end{lemma}

\subparagraph{Computation model.}
Throughout the paper we assume the Random Access Machine model of computation,
where one can allocate an array of size $n$ in $\bigO(n)$ time.
Constant-time operations include accessing and updating an array element,
adding or deleting an element in a linked list,
and the common arithmetic operations on integers.
\section{Enumerating Join Trees}\label{section:enum}

Our strategy for enumerating the join trees of a hypergraph $\hyp$
 starts from~\cref{thm:join_tree_mst}
 which allows us to reduce the problem
 to the enumeration of maximum spanning trees of 
 the line graph $\lin(\hyp)$.
The best known algorithm for MST enumeration\footnote{
 Eppstein focuses on minimum spanning trees, 
 but the same algorithm applies to maximum spanning trees 
 with flipped comparisons.} 
 is due to Eppstein~\cite{eppstein1995representing} by deriving from 
 the input graph $\G$ a so-called {\em equivalent graph} $\EG$. 
Every spanning tree of $\EG$ corresponds to an MST of
$\G$ and vice versa. 
Eppstein gives an algorithm (and proves a matching lower bound) 
 in time $\bigO(m + n \log n)$
 to construct the equivalent graph from an arbitrary weighted graph
 with $m$ edges and $n$ vertices.
Then to enumerate all $k$ MSTs of $\G$, he
 applies existing algorithms to enumerate the spanning trees of $\EG$.
Since there are optimal spanning tree enumeration algorithms
 that run in $\bigO(m + n + k)$ time~\cite{Kapoor1995,DBLP:journals/siamcomp/ShiouraTU97},
 the overall time complexity to enumerate MSTs is $\bigO(m + n \log n + k)$.
The main result of this section is an algorithm for enumerating join trees leveraging the structure
of acyclic hypergraphs and their line graphs. In particular:
\begin{itemize}
    \item Given the line graph $\lin$ of an $\alpha$-acyclic hypergraph, 
    we can construct an equivalent graph of $\lin$ in $\bigO(|\lin|)$, 
    thus enumerating the join trees in $\bigO(|\lin| + k)$ time (\cref{thm:alpha_JTE_enum}).
    \item Given any $\gamma$-acyclic hypergraph $\hyp$, 
    we can construct an equivalent graph of $\lin(\hyp)$ in $\bigO(|\hyp|)$,
    lowering the overall time complexity of enumeration to $\bigO(|\hyp| + k)$ (\cref{thm:gamma_JTE_enum}).
\end{itemize} 
Note that $|\lin(\hyp)|$ can be quadratic in $|\hyp|$
 while $|\hyp| \in \bigO(|\lin(\hyp)|)$,
 so the second item above yields a better bound.

In the rest of this section, we first define what is an equivalent graph. 
Then, we present the algorithm for enumerating join trees of 
$\alpha$-acyclic hypergraphs.
Finally, we adapt the algorithm to enumerate join trees
of $\gamma$-acyclic hypergraphs more efficiently.


\subsection{Equivalent Graph} \label{subsection:EquivalentGraph}
The key idea of Eppstein's algorithm~\cite{eppstein1995representing}
 is to construct an {\em equivalent graph} $\EG$
 whose spanning trees one-to-one correspond to the MSTs of the input graph $\G$,
 thereby reducing MST enumeration to spanning tree enumeration.
\begin{definition}[Equivalent Graph~\cite{eppstein1995representing}]
    \label{def:equivalent_graph}   
\revB{ Given a weighted graph $G = ((\R, \E, \nodes), \w)$, 
    a multigraph $\EG=((\R, \E, \nodesEG), \w)$ is an {\em equivalent graph} of $\G$
    if the spanning trees of $\EG$ one-to-one correspond 
    to the maximum spanning trees of $\G$:
    a set of edges $\E^T \subseteq \E$ induces a spanning tree of $\EG$
    if and only if $\E^T$ induces an MST of $\G$.}
\end{definition}
In what follows, we may use the notation $\EG(\lin(\hyp))$ or $\EG(\hyp)$ to refer to an equivalent graph of the line graph of hypergraph $\hyp$. 
Note that $\EG$ shares the same vertices and edges of $\G$,
 and they only differ in the incidence function
 mapping each edge to its endpoints.
Such an equivalent graph can be constructed from $\G$
 by applying a series of \emph{sliding transformations}:
 for two incident edges $\e, \e'$ with
 $\nodes(\e) = \{\rel_1, \rel_2\}$ and $\nodes(\e') = \{\rel_2, \rel_3\}$,
 we can {\em slide} $\e'$ along $\e$ by updating $\nodes(\e') = \{\rel_1, \rel_3\}$,
 if $\w(\e') < \w(\e)$.
Note that the edge $\e'$ retains its identity after sliding,
 which explains the need for the incidence function $\nodes$.
Furthermore, we will use a given rooted MST $\T$ as a guide, 
 and only slide edges towards the root. More formally:

\begin{definition}[Sliding Transformation~\cite{eppstein1995representing}]\label{def:sliding_transformation}
Let $T$ be a MST of $G$ rooted at $\rel$.
Let $\e$ be an edge such that $\nodes(\e) = \{\rel_1, \rel_2\}$ with $\rel_1$ closer to the root than $\rel_2$ in $T$.
If another edge $\e'$ shares $\rel_2$ with $\e$,
 i.e., $\nodes(\e') = \{\rel_2, \rel_3\}$,
 and $\w(\e') < \w(\e)$, 
 then {\em sliding} $\e'$ along $\e$ 
 results in a graph $\G'$ that is identical to $\G$,
 except that the incidence function maps $\e'$ to $\{\rel_1, \rel_3\}$.
\end{definition}
 
The key result by Eppstein shows that applying sliding transformations
 on $\G$ along a rooted MST $\T$ to a fixpoint
 results in an equivalent graph of $\G$.
To state this formally, we first define an ordering on graphs
 based on sliding transformations.
It is easy to verify the following is a partial order:
\begin{definition}[Sliding Partial Order]
\revB{We write $\G \slide{\T} \G'$ if $\G'$ can be obtained
 from $\G$ by applying a sequence of sliding transformations
 along a rooted MST $\T$ of $\G$.}
\end{definition}
Applying sliding transformations to a fixpoint
 therefore yields a maximal element under the sliding partial order.
\begin{theorem}[Sliding produces $\EG$ at fixpoint~\cite{eppstein1995representing}]\label{thm:eg_max_under_sliding}
    Given a weighted graph $\G$ and a rooted MST $\T$, and let $\EG$
    be a maximal element under $\slide{\T}$. Then $\EG$ is an equivalent graph of $\G$.
\end{theorem}
The choice of the initial rooted MST can affect the structure of the equivalent graph, 
but the order of sliding transformations performed has no impact~\cite{eppstein1995representing}.

\begin{figure}[tbp]
  \centering
  \begin{subfigure}[b]{0.25\textwidth} 
    \centering
    \includegraphics[height=3cm,keepaspectratio]{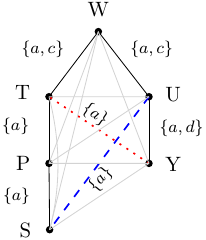}
    \caption{$\T_W$}
    \label{subfig:sliding_transformation_tw}
  \end{subfigure}\hfill
  \begin{subfigure}[b]{0.25\textwidth} 
    \centering
    \includegraphics[height=3cm,keepaspectratio]{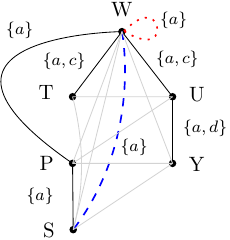}
    \caption{3 Edges Slid}
    \label{subfig:sliding_transformation_3slid}
  \end{subfigure}\hfill
  \begin{subfigure}[b]{0.3\textwidth} 
    \centering
    \includegraphics[height=3cm,keepaspectratio]{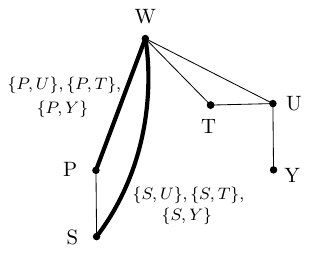}
    \caption{$\EG$}
    \label{subfig:sliding_transformation_eg}
  \end{subfigure}\hfill
  \begin{subfigure}[b]{0.2\textwidth} 
    \centering
    \includegraphics[height=3cm,keepaspectratio]{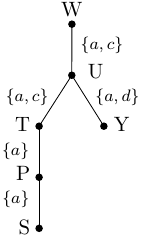}
    \caption{$\T'_W$}
    \label{subfig:sliding_transformation_twprime}
  \end{subfigure}\hfill
  \caption{
    $\T_W$ is an MST of the line graph in \cref{subfig:hyp_lin_jt_l6}
    where the black solid lines stand for the tree edges, the dashed line
    for an MST edge $(S, U)$, the dotted line for a non-MST edge $(T, Y)$
    and gray solid lines for the remaining non-tree edges. $\EG$ is the
    resulting equivalent graph where the thick lines highlight the
    parallel edges.
  }
  \label{fig:sliding_transformation}
  
\end{figure}

\begin{example}\label{example:sliding}
Given the 6-clique line graph $\lin_6$ in \cref{subfig:hyp_lin_jt_l6}, we
find a rooted MST $\T_{W}$ as shown in \cref{subfig:sliding_transformation_tw}.
Among all the tree edges shown as black solid lines, 
we can only apply the sliding transformation
    to the edge $\{T, P\}$ whose weight $\w(\{T, P\}) = |\{a\}| = 1$ is 
    lighter than its parent tree edge $\{W, T\}$ with 
$\w(\{W, T\}) = |\{a, c\}| = 2$.
    We slide along the tree edge $\{W, T\}$ to the root
    so that the edge $\{T, P\}$ becomes $\{W, P\}$ as shown by the solid curve
    in \cref{subfig:sliding_transformation_3slid}.
Non-tree edges can slide similarly. We consider two examples,
    $\{S, U\}$ illustrated with a dashed line and $\{Y, T\}$ 
 with a dotted line in
\cref{subfig:sliding_transformation_tw}.
All other non-tree edges are shown in light gray.
    We can slide $\{S, U\}$ along tree edge $\{U, W\}$ 
    to become $\{S, W\}$ as shown by the dashed curve in 
    \cref{subfig:sliding_transformation_3slid}.
    Both ends of $\{Y, T\}$ can slide along the tree edges to the root
    so that the edge becomes a self-loop
    as shown by the dotted loop in 
    \cref{subfig:sliding_transformation_3slid}.
This edge will not appear in any spanning tree of $\lin_6^\equiv$,
 and therefore not a part of any MST of $\lin_6$.
We refer to such an edge as a {\em non-MST edge}, as opposed to 
an {\em MST edge}.
Denoting the graph after sliding as $\lin_6'$, ordering
    $\lin_6 \slide{\T_{W}} \lin_6'$ holds.
By applying sliding transformations to a fixpoint,
we obtain an equivalent graph $\lin_6^\equiv$, where there are 
two sets of parallel edges highlighted by thick lines
in \cref{subfig:sliding_transformation_eg} (self-loops are omitted). 
    For example, the tree edge $\{T, P\}$ in $\T_{W}$ and
    non-tree edges $\{P, U\}, \{P, T\}, \{P, Y\}$ become parallel in 
$\lin_6^\equiv$ between $P$ and $W$. 
Therefore
$\lin_6 \slide{\T_{W}} \lin_6' \slide{\T_{W}} \lin_6^\equiv$ holds under the sliding partial order.
We can easily verify that each spanning tree of
$\lin_6^\equiv$ corresponds to an MST of $\lin_6$, such as
$\T'_{W}$ in \cref{subfig:sliding_transformation_twprime}.
\end{example}

\subsection{\texorpdfstring{Enumerating Join Trees of $\alpha$-Acyclic Hypergraphs}{Enumerating Join Trees of alpha-Acyclic Hypergraphs}} \label{subsection:JTAplhaAcyc}

 

The bottleneck of Eppstein's algorithm for constructing equivalent graphs
of arbitrary weighted graphs lies in a subroutine
that identifies where each edge will eventually slide to.
Because each edge can only slide along a heavier edge,
 it will eventually be ``blocked'' by a lighter or equally weighted edge
 along its path to the root.
The subroutine essentially performs binary search
 to find the blocking edge,
 leading to the $\log n$ factor in the overall complexity.
The key to our improvement is to show that for every acyclic hypergraph $\hyp$,
 we can construct an {\em equivalent hypergraph} $\hyp^*$
 whose join trees one-to-one correspond to those of $\hyp$,
 but one special join tree of $\hyp^*$ has 
 {\em monotonically increasing weight} from root to leaf,
 which enables constant-time identification of the blocking edge. 
\begin{definition}\label{def:mono_weight_tree}
    A {\em monotonic weight join tree} $\MWJT \in \JTof{\hyp}$
    is a rooted join tree of $\hyp$ such that for any
    $\e\in \E(\MWJT)$ that has a parent edge $\parent(\e)$, 
    $\w(\e) > \w(\parent(\e))$.
\end{definition}

We can always construct such a $\hyp^*$ given any $\alpha$-acyclic hypergraph $\hyp$:

\begin{theorem}\label{thm:MWJT}
    Given an $\alpha$-acyclic hypergraph $\hyp = (\X, \R, \vars)$, 
    there exists an {\em equivalent hypergraph} $\hyp^* = (\X^*, \R, \vars^*)$ over the 
    same hyperedge set $\R$ that
    \begin{itemize}
        \item admits a monotonic weight join tree $\MWJT \in \JTof{\hyp^*}$, and
        \item $\JTof{\hyp^*} = \JTof{\hyp}$.
    \end{itemize}
\end{theorem}

\begin{figure}[tbp]
  \centering
  \begin{subfigure}[b]{0.33\textwidth}
    \centering
    \includegraphics[height=3cm,keepaspectratio]{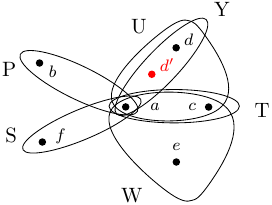}
    \caption{$\hyp^*_6$}
    \label{subfig:mwjt_hs6}
  \end{subfigure}%
  \begin{subfigure}[b]{0.33\textwidth}
    \centering
    \includegraphics[height=3cm,keepaspectratio]{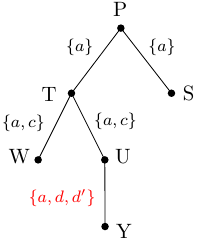}
    \caption{$\MWJT_p$}
    \label{subfig:mwjt_mwjtp}
  \end{subfigure}%
  \begin{subfigure}[b]{0.33\textwidth}
    \centering
    \includegraphics[height=3cm,keepaspectratio]{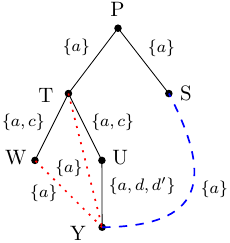}
    \caption{$\MWJT_p$ with 3 non-tree edges}
    \label{subfig:mwjt_mwjtp_non_tree}
  \end{subfigure}

  \caption{
    An equivalent hypergraph $\hyp^*_6$ and its monotonic weight join tree $\MWJT_P$
    rooted at $P$.
  }
  \label{fig:mwjt}
\end{figure}

\begin{example}
$\hyp^*_6$ in \cref{subfig:mwjt_hs6} is a hypergraph with similar structure 
 to $\hyp_6$ in \cref{subfig:hyp_lin_jt_h6}.
$\hyp^*_6$ differs from $\hyp_6$ by one vertex $d'$
and admits the same set of join trees.
Therefore, finding an equivalent graph $\EG(\lin(\hyp^*_6))$
is sufficient for enumerating the join trees of $\hyp_6$.
$\hyp^*_6$ also admits a monotonic weight join tree 
$\MWJT_P$ in \cref{subfig:mwjt_mwjtp}. 
\end{example}

\revB{A monotonic weight join tree $\MWJT$ can accelerate the construction
 of the equivalent graph in several ways.
First, all tree edges are already ``in place'',
 because they cannot slide along their lighter parent edges.
We therefore need only consider non-tree edges.
For each non-tree edge $\e = \{\rel_i, \rel_j\}$, there are two possible cases.
First, if one endpoint, say $\rel_i$, is an ancestor of the other,
 say $\rel_j$,
 then $\e$ can only slide to become parallel to the edge
 right below $\rel_i$ on the path connecting $\rel_i$ and $\rel_j$ in $\MWJT$.
This is because all edges on that path must be no lighter than $\e$
 due to the running intersection property of join trees,
 but since $\MWJT$ has monotonically increasing weight,
 at most one edge (the one right below $\rel_i$) can have equal weight to $\e$.
If this edge has weight equal to $\e$, it becomes the blocking edge,
 and $\e$ slides to become parallel to it;
 otherwise, $\e$ slides to a self-loop at $\rel_i$.
In the second case, if neither endpoint is an ancestor of the other,
 there are at most two blocking edges, 
 namely those below the lowest common ancestor of $\rel_i$ and $\rel_j$
 on the path between them.
We therefore only need to compare the weight of $\e$ with
 these two potentially blocking edges to determine the destination of $\e$.}

\begin{example}
\revB{
Consider again $\MWJT_p$ in~\cref{subfig:mwjt_mwjtp_non_tree}.
First note that every tree edge is heavier than its parent edge,
 and each non-tree edge is no heavier than any edge on the path
 connecting its endpoints in the tree.
There is one edge, $\{T, Y\}$, where one endpoint ($T$) is an ancestor of the other ($Y$).
The only potentially blocking edge is $\{T, U\}$ right below $T$.
But because $\{T, U\}$ is heavier,
 $\{T, Y\}$ slides to become a self-loop at $T$.
For the edge $\{W, Y\}$, all edges on the tree path connecting the endpoints 
 are heavier, so $\{W, Y\}$ also slides to a self-loop at their LCA $T$.
Finally, $\{Y, S\}$ is blocked by the two edges $\{P, T\}$ and $\{P, S\}$
 below their LCA $P$, as all three edges have a weight of 1.
Therefore $\{Y, S\}$ slides to $\{T, S\}$.}
\end{example}

Our algorithm for constructing the equivalent graph
 leverages the insight above to slide each edge in constant time.
In the following, we first define a few helpful notations and
 describe a preprocessing step to build helper data structures,
 before presenting the algorithm.

Given a rooted MCS tree $\T_\rel$ of a line graph $\lin$, a non-tree edge 
is $\e = \{\rel_i, \rel_j\} \in \E(\lin) \setminus \E(\T_\rel)$.
There is a path in $\T_\rel$ between $\rel_i$ and $\rel_j$ via their
lowest common ancestor $\LCA(\rel_i, \rel_j)$.
We define the {\em LCA edges} $\LCAE(\e) = \LCAE(\rel_i, \rel_j)$ 
as a set of at most two tree edges on the path and incident to 
$\LCA(\rel_i, \rel_j)$.
If $\rel_i$, $\rel_j$ are ancestor and child, then
$|\LCAE(\e)| = 1$, otherwise, $|\LCAE(\e)| = 2$.

During pre-processing, we first conduct a breadth-first 
search on $\T_\rel$ to obtain the depth table $\depth$ of each 
tree node in $\bigO(|\T_\rel|)$.
Then we build two data structures
in $\bigO(|\T_\rel|)$ to facilitate
 the constant-time query of the lowest common ancestor (LCA)~\cite{bender2000lca} 
 given two nodes,
 and the level ancestor (LA)~\cite{bender2004level}
 of a node $\rel$, 
 which is the ancestor of $\rel$ at a given depth.
\cref{algo:buildEG} assumes these data structures are prebuilt and available,
 which allows finding $\LCAE(\e)$ of any non-tree edge $\e$ in constant time
 as follows:
\begin{align*}
\LCAE(\{\rel_i, \rel_j\}) = 
\begin{cases}
        \{\{\rel_i, \LA(\rel_j, \depth(\rel_i) + 1)\}\} & \text{if } \rel_i = \LCA(\rel_i, \rel_j) \\
        \{\{\rel_j, \LA(\rel_i, \depth(\rel_j) + 1)\}\} & \text{if } \rel_j = \LCA(\rel_i, \rel_j) \\
        \{\{l, \LA(\rel_i, d)\}, \{l, \LA(\rel_j, d)\}\} & \text{otherwise, } l \defeq \LCA(\rel_i, \rel_j), d \defeq \depth(l) + 1
\end{cases}
\end{align*}

\begin{figure}
  \refstepcounter{figure} 
  \centering
  \begin{minipage}[c]{0.54\textwidth}
    \input{algos/algo_buildEG.tex}
  \end{minipage}
  \hfill
  \begin{minipage}[c]{0.44\textwidth}
    \vspace{.3cm}
    \setcounter{subfigure}{0} 
    \begin{subfigure}[b]{\textwidth}
    \centering
    \includegraphics[height=3.3cm,,keepaspectratio]{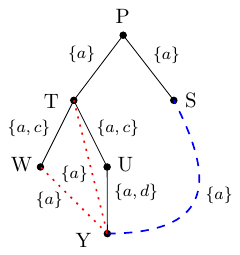}
    \caption{MCS Tree $\T_P(\hyp_6)$}
    \label{subfig:buildEG_tp}
    \end{subfigure}

    
    \begin{subfigure}[b]{\textwidth}
    \centering
    \includegraphics[height=3.3cm,,keepaspectratio]{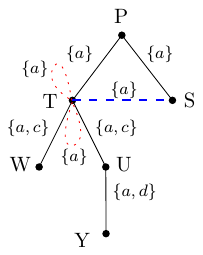}
    \caption{3 Non-tree Edges Processed}
    \label{subfig:buildEG_3slid}
    \end{subfigure}

    \hfill
    \begin{subfigure}[b]{.95\textwidth}
    \centering
    \includegraphics[height=3.3cm,,keepaspectratio]{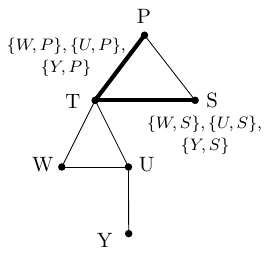}
    \caption{Equivalent Graph $\EG(\hyp_6)$}
    \label{subfig:buildEG_eg}
    \end{subfigure}
    \hfill
    \addtocounter{figure}{-1} 
    \captionof{figure}{\cref{algo:buildEG} on $\hyp_6$ of \cref{subfig:hyp_lin_jt_h6}}
    \label{fig:buildEG}
  \end{minipage}
\end{figure}

\revB{ We are now ready to present \cref{algo:buildEG} which constructs $\EG$
 by sliding each non-tree edge $\e = \{\rel_i, \rel_j\}$ with weight $\w^*$.
There are two cases.
First, if $\LCAE(\e)$ returns a single edge,
 then one endpoint is an ancestor of the other.
By the definition above $\LCAE(\e)$ always returns the ancestor node $l$ first,
 while the other node is denoted as $l'$.
The weight of $\{l, l'\}$ is then compared with $\w^*$:
 if $\w^*$ is lighter, $\e$ slides to a self-loop at $l$;
 otherwise the weights must be equal,
 and $\e$ slides to be parallel to $\{l, l'\}$.
Second, if $\LCAE(\e)$ returns two edges,
 then neither endpoint is an ancestor of the other.
In this case we compare $\w^*$ with the two edges in $\LCAE(\e)$.
If $\w^*$ is lighter than both, $\e$ slides to a self-loop at $l$;
 if $\w^*$ equals one of them, $\e$ slides to be parallel to that edge;
 otherwise, $\w^*$ equals both, and $\e$ slides to be incident to each of the highest depth endpoints of both edges.
The four cases are exhaustive, as the non-tree edge never outweighs
any of its LCA edges by the running intersection property.}

\begin{theorem}\label{thm:algo_buildEG_on_H*}
    \revB{Given a hypergraph $\hyp^* = (\X^*,\R,\vars^*)$ with
    a monotonic weight join tree $\MWJT$, \cref{algo:buildEG}
    produces an Equivalent Graph $\EG(\lin(\hyp^*))$.}
\end{theorem}

The notation $\hyp^*$ highlights later applications to equivalent hypergraphs,
although \cref{thm:algo_buildEG_on_H*} applies to any hypergraph admitting a
monotonic weight join tree.

\revA{Although~\cref{algo:buildEG} constructs an equivalent graph of $\lin(\hyp^*)$
 in constant time per edge, such an $\hyp^*$ may be asymptotically larger than $\hyp$.
Perhaps surprisingly, running~\cref{algo:buildEG}
 {\em directly on $\hyp$} produces the same equivalent graph!
This is because $\hyp^*$ is constructed in a way that preserves the
 structure of $\hyp$, keeping the relative weights of an edge and its LCA edges invariant.
Intuitively, we construct $\hyp^*$ by making ``local copies'' of the vertices
 in $\hyp$: this perturbs the weights to be monotonic along $\MWJT$,
 while ensuring the weights are adjusted consistently.
In particular, when making a new copy $\x'$ of a vertex $\x$,
 we add $\x'$ to {\em all hyperedges} that contain $\x$, 
 namely $\hyp|_\x$.
This way, $\vars(e) \subset \vars(e')$ in $\lin(\hyp)$
 if and only if $\vars(e) \subset \vars(e')$ in $\lin(\hyp^*)$.
As a result, all comparisons in~\cref{algo:buildEG} return the same result 
 when running on $\hyp$ or $\hyp^*$.
We illustrate this with an example, while deferring the detailed construction of $\hyp^*$
 to \cref{algo:buildDG} in \cref{appendix:missing_proofs_enum}.}

\begin{example}
\revA{
In~\cref{subfig:mwjt_hs6}, we made a copy $d'$ of $d$ and added it
 to all hyperedges containing $d$ in $\hyp_6$.
This maintains the relative weights between the tree edge $\{U, Y\}$
 and the non-tree edges highlighted in~\cref{subfig:mwjt_mwjtp_non_tree}:
 all three non-tree edges remain lighter than $\{U, Y\}$.
Were we to make a copy of $a$ to add to $\hyp_6|_a$ 
 (which contains every hyperedge in $\hyp_6$),
 the weights of the non-tree edges
 relative to their respective LCA edges would also remain unchanged.}
\end{example}

Finally, because $\hyp$ and $\hyp^*$ share the same join trees,
 there must be $\EG(\lin(\hyp))$ that is the same as $\EG(\lin(\hyp^*))$.
Together with~\cref{thm:algo_buildEG_on_H*}, this means running~\cref{algo:buildEG}
 directly on $\lin(\hyp)$ produces an equivalent graph $\EG(\lin(\hyp))$ of $\lin(\hyp)$:

\begin{theorem}\label{thm:algo_buildEG_on_H}
    \revB{Given the line graph $\lin$ and an MCS tree $\T_\rel$ of an
    $\alpha$-acyclic hypergraph $\hyp$,
    \cref{algo:buildEG} returns an equivalent graph $\EG(\lin)$
    in time $\bigO(|\lin|)$.}
\end{theorem}

Once we compute the equivalent graph $\EG(\lin)$, we can apply existing 
 algorithms~\cite{Kapoor1995,DBLP:journals/siamcomp/ShiouraTU97} to enumerate all its spanning trees by edits
 with amortized constant delay.
And because the MCS tree can be constructed in linear time
 from $\hyp$ which in turn can be recovered from $\lin(\hyp)$
 in linear time, the overall time complexity of enumerating all join trees of $\hyp$ is
 linear in the size of $\lin$ plus the number of join trees:

\begin{theorem}\label{thm:alpha_JTE_enum}
        \revB{Given the line graph $\lin$ of an $\alpha$-acyclic hypergraph $\hyp$,
        the join trees of $\hyp$ can be enumerated in time $\bigO(|\lin| + |\JTof{\hyp}|)$.}
\end{theorem}




\subsection{\texorpdfstring{$\gamma$-Acyclic Queries}{gamma-Acyclic Queries}} \label{subsection:JTGammaAcyc}

%
The run time of~\cref{algo:buildEG} depends on the size of
the line graph which can be quadratically larger than
the input hypergraph.
\revA{
If $\hyp$ is $\gamma$-acyclic,
 we can bring the total time complexity of enumeration down to
$\bigO(|\hyp| + |\JTof{\hyp}|)$.
Using an algorithm by Leitert~\cite{leitert2021ujg},
 we can construct $\lin(\hyp)$ from $\hyp$ in time
$\bigO(|\hyp| + |\lin(\hyp)|)$ when $\hyp$ is $\gamma$-acyclic;
 furthermore, because every edge in $\lin(\hyp)$ 
 of a $\gamma$-acyclic $\hyp$ is an MST edge~\cite{leitert2021ujg},
 we have $|\hyp| + |\JTof{\hyp}| \in \Omega(|\lin(\hyp)|)$,
 and so $\bigO(|\hyp| + |\lin(\hyp)| + |\JTof{\hyp}|) = \bigO(|\hyp| + |\JTof{\hyp}|)$,
 hiding the $|\lin(\hyp)|$ term in the overall complexity.}
\revB{
However, the $\lin(\hyp)$ constructed by Leitert is {\em unweighted},
 so we need to modify \cref{algo:buildEG} slightly to slide each non-tree edge
 based on the weights of its LCA edges alone.
We first modify MCS (\cref{algo:MCS}) to track the weight of each tree edge
 as in~\cref{algo:MCS4gamma}.
Then, we update the cases of~\cref{algo:buildEG}
 to those in~\cref{algo:buildEG4gamma}.
First, if $\LCAE(\e)$ returns a single edge,
 we directly slide $\e$ to be parallel to it,
 because we know $\e$ is an MST edge and therefore will not become a self-loop.
Otherwise, if $\LCAE(\e)$ returns two edges,
 we {\em compare the weights of those two edges}, 
 instead of comparing them with $\w(\e)$.
This is again because we know $\e$ is an MST edge,
 so it must have weight equal to at least one of its LCA edges.
If one of the LCA edges is lighter,
 $\e$ must have equal weight to that edge,
 and we slide $\e$ to be parallel to it.
Otherwise, if both LCA edges have the same weight,
 $\e$ must also have the same weight,
 and we slide $\e$ to be under both edges to form a triangle.
In all cases we only use the weights of the tree edges
 and do not need to know $\w(\e)$.}

\begin{figure}[tbp]
  \centering
  \begin{minipage}[t]{0.43\textwidth}
    \input{algos/algo_MCS4gamma.tex}
  \end{minipage}
  \hfill
  \begin{minipage}[t]{0.51\textwidth}
    \input{algos/algo_buildEG4gamma.tex}
  \end{minipage}
\end{figure}

\begin{theorem}\label{thm:gamma_JTE_enum}
    \revB{
    The join trees of a $\gamma$-acyclic $\hyp$
     can be enumerated in time
    $\bigO(|\hyp| + |\JTof{\hyp}|)$.}
\end{theorem}


\section{The Canonical Join Tree of a Berge-Acyclic Query}\label{section:cjt}

An acyclic query can have exponentially many join trees
 with respect to its size.
For example, the line graph of a clique query with $n$ relations 
 is an $n$-clique $K_n$ with $n^{n-2}$ join trees
 by Cayley's formula~\cite{cayley1889trees}.
Enumerating all join trees can be prohibitive for large queries.
The query optimizer does not need to consider all 
 possible join trees to achieve good performance.
For example, the implementation of Yannakakis' algorithm by 
Zhao et al.~\cite{zhao2025debunking} 
achieves similar performance on any join tree
rooted at the largest relation. 
An alternative to enumeration is therefore
 to simply construct one join tree for a given root.
This can be done in linear time by the Maximum Cardinality
 Search (MCS) algorithm from a chosen relation as shown in \cref{algo:MCS}.
In this section, we prove that for a Berge-acyclic query 
the MCS algorithm produces a {\em shallowest tree},
 where the depth of each tree node is minimized.
We prove that this shallowest tree is unique, and therefore call it 
the {\em canonical join tree}.

\begin{definition}\label{def:canonical_tree}
A join tree $\T_\rel$ rooted at $\rel$ is {\em canonical} if
$\depth(\T_{\rel}, \rel_i) \leq \depth(\T'_{\rel}, \rel_i)$
for any other join tree $\T'_\rel$ rooted at $\rel$
and any $\rel_i \in \R(\T_\rel) = \R(\T'_\rel)$.
\end{definition}
A shallow join tree has practical benefits.
For example, the depth of the join tree determines the number of
 sequential steps required in a parallel join algorithm.
A shallow join tree also tends to be wide and have more leaves,
 allowing better utilization of indices.

\revC{
Although Berge-acyclicity was thought to be too restrictive when it was
 first introduced to database theory~\cite{fagin1983degrees},
 we found it to be general enough to cover almost all acyclic queries
 encountered in the wild.
As shown in \cref{tab:experiments_benchmarks}, 
among \num{10454} queries from five popular benchmarks, \num{9285} are 
$\alpha$-acyclic, and only 8 of these are not Berge-acyclic.
In retrospect, this should not be surprising,
 as most joins in relational databases are over primary/foreign keys.
 Emerging workloads in graph databases usually involve simple graphs
 and seldom require composite key joins.
A query without composite key joins admits a {\em linear} hypergraph,
 where each pair of hyperedges shares at most one vertex.
The following result establishes an equivalence between
 $\alpha$-acyclicity with linearity and Berge-acyclicity.}

\begin{table}[tbp]
    \centering
\begin{tabular}{|c|c|c|c|c|}
\hline
Name       & \# Queries & \# $\alpha$-Acyclic & \# Composite-Key Joins & \# Berge-Acyclic \\ \hline
TPC-H\cite{tpch}      & 22            & 21      & 2      & 19                   \\ \hline
JOB\cite{leis2015good_JOB_benchmark}        & 113           & 113       &  0          & 113                   \\ \hline
STATS-CEB\cite{DBLP:journals/pvldb/HanWWZYTZCQPQZL21_stats_CEB_benchmark}  & 2603          & 2603       & 0        & 2603                   \\ \hline
CE\cite{chen2022accurate_CE_benchmark}         & 3004          & 1839       & 0        & 1839                   \\ \hline
Spider-NLP\cite{yu2018spider_nlp_benchmark} & 4712          & 4709      & 6       & 4703                   \\ \hline
\end{tabular}
\caption{\revC{Acyclic queries in the benchmarks (all $\alpha$-acyclic queries are also $\gamma$-acyclic.)}}
\label{tab:experiments_benchmarks}
\end{table}

\begin{proposition}\label{prop:berge_is_no_composite_alpha}
    An $\alpha$-acyclic hypergraph is Berge-acyclic if and only if it is linear.
\end{proposition}

By \cref{prop:berge_is_no_composite_alpha}, 
 every edge in the line graph of a Berge-acyclic hypergraph
 has a weight of 1.
Every spanning tree is a maximum spanning tree, therefore a join tree.

\begin{corollary}\label{cor:spanning_tree_join_tree}
For a Berge-acyclic hypergraph $\hyp$, any spanning tree of $\lin(\hyp)$
is a join tree.
\end{corollary}

In the rest of this section, we prove the existence and uniqueness 
 of the canonical join tree rooted at any relation of a Berge-acyclic hypergraph,
 and show that it can be constructed by MCS as in \cref{algo:MCS}.

\revB{The key insight leading to the existence and uniqueness of the canonical
join tree is that the line graph $\lin$ of a Berge-acyclic hypergraph is 
{\em geodetic}~\cite{ore1962theory},
 meaning that there is a unique shortest path between any pair of vertices.
The canonical join tree can then be constructed
 by taking the union of all shortest paths from the root to each other vertex,
 as this will guarantee minimal depth for each vertex.
The key step to establish the geodetic property is to show that
 $\lin$ is a special class of {\em chordal graphs} called {\em block graphs}.
Chordal graphs are intimately related to acyclic hypergraphs~\cite{d1988hypergraph},
 and a block graph is a special chordal graph defined as follows:}
\begin{definition}[Chordal and Block]\label{def:block}
A simple graph $\G = (\R, \E)$ is 
\begin{itemize}
    \item {\em chordal}~\cite{DECARIA2016261} if every cycle of length at least 4 has a chord, 
    i.e., an edge that is not part of the cycle but connects two vertices 
    of the cycle;
    \item a {\em block graph}~\cite{harary1963characterization} if it is chordal
    and {\em diamond-free}, i.e., no subgraph induced by any $\R' \subseteq \R$ is a diamond (\cref{def:diamond}).
\end{itemize}

\end{definition}

\begin{lemma}\label{lemma:berge_block_graph}
The line graph $\lin$ of a Berge-acyclic hypergraph $\hyp$ is a block graph.
\end{lemma}

%
Together with the fact that every block graph is {\em geodetic} 
(with a unique shortest path between any two vertices)~\cite{ore1962theory}, 
\cref{lemma:berge_block_graph} implies the following corollary.
\begin{corollary}\label{cor:unique_shortest_path}
Let $\lin$ be the line graph of a Berge-acyclic hypergraph $\hyp$.
There is a unique shortest path between any two vertices in $\lin$.
\end{corollary}
We are now ready to prove the existence and uniqueness of the canonical join tree.
\begin{theorem}\label{thm:canonical_join_tree}
\revB{A Berge-acyclic hypergraph $\hyp$ has a unique canonical tree.}
\end{theorem}
\begin{proof}
Let $\lin$ be the line graph of $\hyp$, 
 and $\spath(\rel, \rel')$ be
 the shortest path in $\lin$ between $\rel, \rel' \in \R(\lin)$,
 we prove that
$\T_{\rel} = \bigcup_{\rel' \in \R(\lin)}\spath(\rel, \rel')$
 is the unique canonical join tree for $\hyp$ rooted at $\rel$.

By~\cref{cor:spanning_tree_join_tree} any spanning tree of $\lin$ is a join tree.
$\T_\rel$ is connected and spans all vertices in $\R(\lin)$,
 because it contains the shortest path from $\rel$ to every $\rel' \in \R(\lin)$.
It remains to show that $\T_\rel$ is acyclic,
 which we prove by induction on the distance $\dist(\rel, \rel')$ 
between $\rel$ and $\rel'$.
Let $\R_d = \setof{\rel' \in \R(\lin)}{\dist(\rel, \rel') \leq d}$.
$\R_0 = \{\rel\}$ contains only the root.
The subgraph $\T_\rel|_{\R_0}$ is trivially acyclic.
Assuming that $\T_\rel|_{\R_{d>0}}$ is acyclic, 
we consider a vertex $\rel' \in \R_{d+1} \setminus \R_d$.
\cref{cor:unique_shortest_path} guarantees a unique shortest path 
between each pair of vertices $\rel, \rel' \in \R(\lin)$.
Each $\rel'$ is connected to a unique neighbor 
$\rel'' \in \R_{d} \setminus \R_{d - 1}$ that is at distance $d$ from $\rel$.
 Otherwise, there are at least two distinct shortest paths from $\rel$ to $\rel'$.
Therefore $\T_\rel|_{\R_{d + 1}}$ is acyclic,
and $\T_\rel$ is a spanning tree thus a join tree of $\hyp$.

The join tree $\T_\rel$ is canonical, because the path from $\rel$ to 
each $\rel' \in \R(\lin)$ is the shortest and therefore minimizing the depth 
$\depth(\T_{\rel}, \rel')$.
The canonical tree is unique by~\cref{cor:unique_shortest_path}.
\end{proof}

The canonical join tree can be constructed by the MCS algorithm
 as shown in~\cref{algo:MCS}:
 
\begin{theorem}\label{thm:berge_cjt}
    Given a Berge-acyclic hypergraph $\hyp$,
    running \cref{algo:MCS} from $\rel \in \R(\hyp)$ constructs the 
    canonical tree $\T_{\rel}(\hyp)$.
\end{theorem}
%
%
\section{Converting a Binary Join Plan to a Join Tree}\label{section:frombinplan}

\begin{figure}
  \begin{minipage}[c]{0.5\textwidth}
    \input{algos/algo_binary2tree.tex}
  \end{minipage}
  \hfill
  \begin{minipage}[c]{0.49\textwidth}
    \vspace{1.1cm}
    \input{figures/b2jt_example.tex}
    \vspace{.2cm}
    \caption{\revA{\cref{algo:bin2tree}\cite{hu2024treetracker} on 
  JOB-3a\cite{leis2015good_JOB_benchmark}}}\label{fig:b2jt_example}
  \end{minipage}
\end{figure}

Recent approaches~\cite{hu2024treetracker,DBLP:journals/pvldb/BekkersNVW25}
that convert a binary join plan into a join tree have gained
popularity as they allow system builders to leverage existing query optimizers
designed for binary join plans.
In this section, we focus on an algorithm by Hu et al.~\cite{hu2024treetracker}
 to convert left-deep linear join plans into join trees as shown in \cref{fig:b2jt_example}.
We prove that the algorithm converts any connected 
left-deep linear join plan into a join tree if and only if the 
query is $\gamma$-acyclic.
This can be seen as a new characterization
 of $\gamma$-acyclic queries.
We formally define binary join plans and describe the algorithm
 by Hu et al.\ in \cref{algo:bin2tree}.

\begin{definition}
A {\em left-deep linear plan} is a sequence of \revB{hyperedges} $(\rel_1, \rel_2, \ldots, \rel_n)$.
It is {\em connected} if
 for each $\rel_{i \geq 2}$, 
 $\exists \rel_{j < i} : \rel_i \cap \rel_j \neq \emptyset$.
\end{definition}

Query optimizers strive to produce connected plans,
 to avoid expensive Cartesian products.
Many optimizers produce exclusively left-deep linear plans.
Plans that are not left-deep are called {\em bushy},
 and such plans may still be decomposed into left-deep fragments~\cite{wang2023free}.

Given a left-deep linear plan, Hu et al.~\cite{hu2024treetracker}
 generate a join tree with \cref{algo:bin2tree}.
\revA{The algorithm chooses the first relation $\rel_1$ as the root and
  iterates through the rest of the plan.
For each relation $\rel_{i \in [2, n]}$,
 it finds the first relation $\rel_j$ that contains all attributes
 shared by $\rel_i$ with all previous relations, $\text{key} \defeq \rel_i \cap \bigcup_{k < i} \rel_{k}$,
 and assigns $\rel_j$ as the parent of $\rel_i$.
The algorithm constructs a join tree
 if it finds a parent for each $\rel_i$.}

Hu et al.~\cite{hu2024treetracker} proved 
that the algorithm succeeds whenever the input plan
 is the reverse of a GYO-reduction order.
They also observed that every left-deep linear plan
 produced for queries in standard benchmarks
 is indeed the reverse of a GYO-reduction order.
This is not a coincidence,
 as we show that every connected left-deep linear join plan
 must be the reverse of a GYO-reduction order
 if and only if the query is $\gamma$-acyclic.
\begin{theorem}\label{thm:bin2tree}
A query is $\gamma$-acyclic
 if and only if every connected left-deep linear join plan
 for the query is the reverse of a GYO-reduction order.
\end{theorem}
Immediately following \cref{thm:bin2tree}, we can conclude the following.
\begin{corollary}\label{cor:bin2tree}
  \revB{For any $\gamma$-acyclic query,
  \cref{algo:bin2tree} converts a given connected left-deep linear join plan
  to a join tree.}
\end{corollary}

\section{Conclusion and Future Work}\label{section:conclusion}

We proposed three approaches for constructing join trees. 
Our enumeration algorithm in~\cref{section:enum} generates join trees 
by edits with amortized constant delay;
in~\cref{section:cjt}, we showed that the Maximum Cardinality Search 
algorithm constructs the unique shallowest join tree for any 
Berge-acyclic query;
in~\cref{section:frombinplan},
 we characterize the class of binary join plans
 that can be converted to join trees. 
\revA{Practitioners can choose from the three approaches to integrate
 instance-optimal algorithms into their system:
 a cost-based optimizer can use our enumeration algorithm
 to generate candidate query plans;
 a system aiming to support very large queries can implement
 the MCS algorithm to generate shallow join trees
 to improve parallelism; 
 and a more conservative extension can derive join trees
 using existing optimization infrastructure,
 while our final result guarantees the validity of the output. }

Future work includes compact representations of join trees for dynamic 
programming, as in binary plan optimizers, and the challenging 
cost estimation for Yannakakis-style algorithms:
the random-walk approach~\cite{DBLP:conf/sigmod/0001WYZ16} models 
joint probabilities for binary joins, 
whereas an efficient and accurate solution for semijoins 
remains to be found.

\begin{figure}[tbp]
  \centering
  \begin{subfigure}[b]{0.2\textwidth} 
    \centering
    \includegraphics[height=2.6cm,,keepaspectratio]{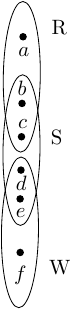}
    \caption{$\hyp_3$}
    \label{subfig:non_berge_unique_h}
  \end{subfigure}\hfill
  \begin{subfigure}[b]{0.2\textwidth} 
    \centering
    \includegraphics[height=2cm,,keepaspectratio]{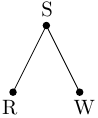}
    \caption{$\T_S(\hyp_3)$}
    \label{subfig:non_berge_unique_t}
  \end{subfigure}\hfill
  \begin{subfigure}[b]{0.2\textwidth} 
    \centering
    \includegraphics[height=2.6cm,,keepaspectratio]{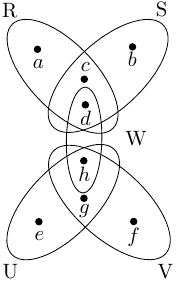}
    \caption{$\hyp_5$}
    \label{subfig:gamma_non_unique_h}
  \end{subfigure}\hfill
  \begin{subfigure}[b]{0.2\textwidth} 
    \centering
    \includegraphics[height=2.6cm,,keepaspectratio]{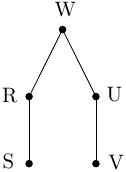}
    \caption{$\T_W(\hyp_5)$}
    \label{subfig:gamma_non_unique_t1}
  \end{subfigure}\hfill
  \begin{subfigure}[b]{0.2\textwidth} 
    \centering
    \includegraphics[height=2.6cm,,keepaspectratio]{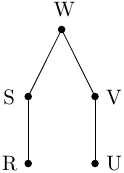}
    \caption{$\T'_W(\hyp_5)$}
    \label{subfig:gamma_non_unique_t2}
  \end{subfigure}\hfill
  \caption{
    $\hyp_3$ is not Berge-acyclic but admits a unique canonical join tree $\T_S(\hyp_3)$
    with any relation chosen as root, such as $\T_S(\hyp_3)$.
    $\hyp_5$ is $\gamma$-acyclic and does not admit unique canonical join trees
    at any relation chosen as root. For example,
    $\T_W(\hyp_5)$ and $\T'_W(\hyp_5)$ are MCS trees generated by \cref{algo:MCS}.
    Neither of them is a canonical join tree rooted at $W$.
  }
  \label{fig:non_berge_n_gamma}
  
\end{figure}

Our work also raises further theoretical questions. 
Can join tree enumeration achieve worst-case constant delay? 
We proved {\em Berge-acyclicity} sufficient for the existence and uniqueness 
of the canonical join tree, but it is not necessary, and 
{\em $\gamma$-acyclicity} is insufficient as shown in 
\cref{fig:non_berge_n_gamma}. 
What is the precise characterization of hypergraphs that admit a unique 
canonical join tree for any root, or for {\em some} root?
\revA{
How can our algorithms be extended to enumerate tree decompositions 
\`a la Carmeli et al.~\cite{DBLP:conf/pods/CarmeliKK17}?
One possible direction is to further develop connections between MSTs of the line graph
 and tree decompositions of the hypergraph.
For example, we are considering using the determinants of the line graph Laplacian
 as heuristics to guide the search for good tree decompositions. }

\bibliographystyle{plainurl}  
\bibliography{references}     

\appendix

\section{Data Structures for Constant Time Operations on Graphs}
\label{appendix:data_structures}


Given line graph $\lin$ and MCS tree $T$, \cref{algo:buildEG} iterates over non-tree edges of $\lin$. 
Moreover, the algorithm uses data structures $\depth(\cdot)$, $\LA(\cdot)$ and $\LCA(\cdot)$, 
in determining the LCA edges of a non-tree edge. 
In this section, we describe how to enumerate over non-tree edges 
in time linear in the number of vertices and edges of $\lin$ and 
how to determine the LCA edges of a non-tree edge in constant time. 

Given a graph $G = (V, E)$,
we assume the vertices are represented by integers from 1 to $|V|$.
Given a tree $\T$ of $G$ rooted at node $\rel$, 
data structure $\depth_\T(\cdot)$ stores for each node $\rel_i$ 
the length of the shortest path between $\rel$ and $\rel_i$ in $\T$. 
Data structure $\LCA_\T(\cdot)$ stores for each pair of nodes 
$\rel_i,\rel_j$ the lowest common ancestor of $\rel_i,\rel_j$ in $\T$. 
Finally, data structure $\LA_\T(\cdot)$ stores for each node-integer pair 
$\rel_i, j$, the ancestor of node $\rel_i$ at depth $j$. 
Each of these data structures can be established from a tree $\T$ in time linear 
in the size of the tree and support constant time look-ups. 

\begin{observation}\label{obs:datastruc}
    Given rooted tree $\T$, data structures 
    $\depth(\cdot)$, $\LCA(\cdot)$ and $\LA(\cdot)$ 
    can be established in time $\bigO(|\T|)$ and 
    support $\bigO(1)$ time look-ups~\cite{CLRS4, SedgewickWayne2011,bender2000lca, bender2004level}.
\end{observation}

Given the line graph $\lin$ of an $\alpha$-acyclic hypergraph $\hyp$ and an MCS tree $\T$ of $\hyp$, recall that the LCA edges $\lambda(e)$ of a non-tree edge $\e 
= (\rel_i,\rel_j)$  of $\T$ is defined as the set of edges incident to $\LCA(\rel_i,\rel_j)$ that lie on the unique path between $\rel_i$ and $\rel_j$ in $\T$. The following observation shows that the LCA edges of a non-tree edge can be found in constant time given data structures $\depth(\cdot)$, $\LCA(\cdot)$ and $\LA(\cdot)$. 

\begin{observation}\label{obs:findLCAedges}
    Given the line graph $\lin$ of an $\alpha$-acyclic hypergraph $\hyp$, an MCS tree $\T$ and data structures $\depth(\cdot)$, $\LCA(\cdot)$ and $\LA(\cdot)$, the LCA edges of any non-tree edge can be found in constant time.
\end{observation}
\begin{proof}
    Given a non-tree edge $\e = (\rel_i,\rel_j)$ we find $\LCA(\rel_i,\rel_j) = l$ in constant time by \cref{obs:datastruc}. The depth $d = \depth(l) + 1$ is likewise found in constant time. The edge $\e_1 = (l,\rel_1)$, where node $\rel_1 = \LA(\rel_i,d)$ is found by a constant time look-up, is adjacent to $l$ and lies on the unique path between $\rel_i$ and $\rel_j$. The same holds for edge $\e_2 = (l,\rel_2)$ where node $\rel_2 = \LA(\rel_j,d)$.
\end{proof}

In the remainder of this section,
 we describe how \cref{algo:buildEG} iterates over and performs sliding transformations on non-tree edges.

We use an array $\lin$ to represent the adjacency list of a line graph where entry $\lin[i]$ consists of a linked list of tuples  $(j, \w)$
 where $j$ is a neighbor of $i$ and $\w$ is the weight of the edge $(i,j)$.
A rooted tree $\T_\rel$ is likewise implemented as an array, 
where the index $i$ represents the vertex $i$,
 and the $i$-th entry $\T_\rel[i]$ stores $(\parent, \w)$ 
 where $\parent$ is the parent of $i$
 and $\w$ is the weight of the edge $(i, \parent)$.
The depth table of the tree is also implemented as an array to enable constant time
lookup for the depth of a vertex $\depth(\T_\rel, i)$.

As we iterate over each edge $(j, \w)$ of $\lin[i]$, we check 
 if the current edge  $(i, j)$ is a tree edge by performing constant-time query of $\T_\rel[i]$
 and $\T_\rel[j]$ to see if one of them is the parent of the other.
From \cref{obs:findLCAedges} we know that we can find the LCA edges $\LCAE(i,j) =\{\e_1,\e_2\}$ in constant time.
There's no general solution to lookup an entry in a linked list in constant time.
 Instead of iterating over the linked-list representation of $\lin$
 to find the edge weights of $\e_1$ and $\e_2$, 
 we can lookup their
 weights from the array representation of $\T_\rel$ in constant time.
Let $\e_1 = \{\rel_1, \parent(\rel_1)\}$ and
 $\e_2 = \{\rel_2, \parent(\rel_2)\}$ with weights $\w_1$ and $\w_2$.
Because both $\e_1$ and $\e_2$ are tree edges, we can find their weights
  by querying $\T_\rel[\rel_1] = (\parent(\rel_1), \w_1)$ 
  and $\T_\rel[\rel_2] = (\parent(\rel_2), \w_2)$. 
Each of these operations takes constant time.

\section{\texorpdfstring{Missing Proofs in \cref{section:preliminaries}}{Missing Proofs in Preliminaries}}
\label{appendix:missing_proofs_prelim}

We prove \cref{lemma:MCS_rip} of section \cref{section:preliminaries} below.
\par\medskip
\noindent{\bf \cref{lemma:MCS_rip}}
{\it 
Let $\T$ be an MCS tree. For an edge $\e \in \T$ that has a parent, 
then
\begin{enumerate}
    \item $\e \not \subseteq \parent(\e)$
    \item For another edge
    $
    \e' \in \T:
        (\e \setminus \parent(\e))\
        \cap\
        (\e' \setminus \parent(\e'))\
        \neq\ \emptyset 
        \quad \implies \quad 
        \parent(\e) = \parent(\e').
    $
\end{enumerate}
}

\begin{proof}
    In~\cref{algo:MCS}, a child $\rel_c$ is only connected to 
    a parent $\rel$ if they share some previously unmarked
    vertex $\x$.
    Because every vertex is marked only once,
    and the parent $\rel_p$ of $\rel$ was labeled
    before $\rel$,
    $\x$ cannot be marked by $\rel$'s parent $\rel_p$.
    Let $\e$ be edge between $\rel$ and $\rel_c$,
    then the edge between $\rel$ and $\rel_p$ is $\parent(\e)$,
    and $\x \in \e$ but $\x \notin \parent(\e)$, 
    therefore $\e \not \subseteq \parent(\e)$.

    We prove the second statement by contrapositive, assuming $\parent(\e) \neq \parent(\e')$.
    We also treat each edge as a set of vertices below, as we can identify each tree
    edge with its set of incident vertices.
    \begin{description}
        \item[Case 1:] $\parent(\e') = \e$, then
        $
        (\e \setminus \parent(\e))
        \cap
        (\e' \setminus \e)
        \subseteq
        (\e \cap \e')
        \setminus
        \e = \emptyset
        $.
        \item[Case 2:]$\parent(\e') \neq \e$. Without loss of generality,
        we suppose $\e$ is no further from the root than $\e'$.
        Let $\parent(\e') = \e''$, then $\e''$ is on the path between 
        $\e$ and $\e'$ in $\T$.
        If 
        $
        (\e \setminus \parent(\e))
        \cap
        (\e' \setminus \e'')
        \neq \emptyset
        $,
        then $\exists\ \x \in (\e' \cap \e) \setminus \e''$,
        which violates the running intersection property of join trees.
    \end{description}
\end{proof}

\section{\texorpdfstring{Missing Proofs in \cref{section:enum}}{Missing Proofs in Section Enumeration}}
\label{appendix:missing_proofs_enum}
In this section we present the missing proofs of \cref{section:enum}. In particular, in \cref{app:proofs_section4.2} we show that there exists an equivalent hypergraph $\hyp^*$ of a given hypergraph $\hyp$ that admits a monotonic weight join tree and preserves the set of join trees of $\hyp$. This is proven in \cref{thm:MWJT}. Thereafter, we prove \cref{thm:algo_buildEG_on_H*} showing that \cref{algo:buildEG} produces an equivalent graph when given input $\lin(\hyp^*)$ and a monotonic weight join tree of $\hyp^*$. Next we show that \cref{algo:buildEG} behaves identically on the input of line graph $\lin(\hyp)$ and an MCS tree of $\hyp$ as it does on $\lin(\hyp^*)$ and a monotonic weight join tree of $\hyp^*$. \cref{algo:buildEG} therefore produces an equivalent graph of $\lin(\hyp)$. This proven in \cref{thm:algo_buildEG_on_H}. Finally, we prove \cref{thm:alpha_JTE_enum}, showing that the join trees of an $\alpha$-acyclic hypergraph can be enumerated in time $\bigO(|\lin| + |\JTof{\hyp}|)$ given its line graph $\lin$.

In \cref{app:proofs_section4.3} we prove \cref{thm:gamma_JTE_enum} showing that the join trees of a $\gamma$-acyclic $\hyp$
     can be enumerated in time
    $\bigO(|\hyp| + |\JTof{\hyp}|)$.

\subsection{\texorpdfstring{Missing Proofs in \cref{subsection:JTAplhaAcyc}}{Missing Proofs in Section Enumeration}}\label{app:proofs_section4.2}

We define the notion of an equivalent hypergraph below.

\begin{definition}
    Given an $\alpha$-acyclic hypergraph $\hyp = (\X, \R, \vars)$, 
    hypergraph $\hyp^* = (\X^*, \R, \vars^*)$, over the 
    same hyperedge set $\R$, is an {\em equivalent hypergraph} of $\hyp$, if $\hyp^*$,
    \begin{itemize}
        \item admits a monotonic weight join tree $\MWJT \in \JTof{\hyp^*}$, and
        \item $\JTof{\hyp^*} = \JTof{\hyp}$.
    \end{itemize}
\end{definition}



We formalize the relationship
 between an $\alpha$-acyclic hypergraph $\hyp$ and 
 its equivalent hypergraph $\hyp^*$ using the concept of {\em hypergraph homomorphisms}.
Note the definition of homomorphisms is not standardized 
 in the literature,\footnote{\revC{A textbook on hypergraphs~\cite{10.5555/2500991} contains two
 different definitions of hypergraph homomorphism;
 one is the same as ours, while the other is also used
 by Scheidt and Schweikardt~\cite{DBLP:conf/mfcs/ScheidtS23}.}} and we choose ours carefully to simplify
 the proofs.




\begin{definition}[Hypergraph Homomorphism]
Let $\hyp_1$ and $\hyp_2$ be hypergraphs.
A \emph{homomorphism} from $\hyp_1$ to $\hyp_2$ is a pair of functions $(\func_{\X}: \X(\hyp_1) \rightarrow \X(\hyp_2),
 \func_{\R}: \R(\hyp_1) \rightarrow \R(\hyp_2))$,
\revC{such that for every hyperedge $\rel_1 \in \R(\hyp_1)$ and every vertex
$\x_1 \in \rel_1$, the image $\func_{\X}(\x_1)$ belongs to the hyperedge
$\func_{\R}(\rel_1)$ in $\hyp_2$.}

We write $\hyp_1 \rightarrow \hyp_2$ to denote the existence of such a homomorphism.
\end{definition}

When there is no ambiguity from context, we simply write $\func$ for
$\func_{\X}$ or $\func_{\R}$.
Intuitively, homomorphisms preserve relations.
A {\em strong} homomorphism also {\em reflects} relations:
%



\begin{definition}[Strong Homomorphism]\label{def:strong_homomorphism}
Let $\func : \hyp_1 \rightarrow \hyp_2$ be a hypergraph homomorphism.
\revC{
We call $\func$ \emph{strong} if it also reflects non-incidence, in the following
sense: for every hyperedge $\rel_1 \in \R(\hyp_1)$ and every vertex
$\x_1 \in \X(\hyp_1)$ that does not belong to $\rel_1$, the vertex
$\func_{\X}(\x_1)$ does not belong to the hyperedge $\func_{\R}(\rel_1)$
in $\hyp_2$.}

We write $\hyp_1 \twoheadrightarrow \hyp_2$ to denote the existence of a strong
homomorphism from $\hyp_1$ to $\hyp_2$.
\end{definition}

Equivalently, under a strong homomorphism, a vertex belongs to a hyperedge in
$\hyp_1$ if and only if its image belongs to the image of that hyperedge in
$\hyp_2$.

We only consider hypergraphs over the same set of
edges,\footnote{%
 The hypergraphs share the same hyperedge set $\R$ 
 but differ in incidence function $\vars$, 
 i.e., the same edge may be incident to different vertices in 
 different hypergraphs.} 
 and from now on
 we assume $\func_\R$ is the identity function.

We show that if there is a strong homomorphism $\func: \hyp' \twoheadrightarrow
\hyp$, then every join tree of $\hyp$ is also a join tree of $\hyp'$.

\begin{lemma}\label{lemma:any_tree_in_H_prime}
    If $\hyp' \twoheadrightarrow \hyp$ then $\JTof{\hyp} \subseteq \JTof{\hyp'}$.
\end{lemma}

\begin{proof}
Let $\T$ be a join tree of $\hyp$, and $\x'$ be any vertex in $\hyp'$. Consider
the neighborhood $\hyp'|_{\x'}$ of $\x'$ in $\hyp'$ comprising the set of hyperedges 
$\setof{\rel' \in \R(\hyp')}{\x' \in \rel'}$.
Since $\hyp' \twoheadrightarrow \hyp$, we have that
 $\x' \in \rel'$
if and only if $\func(\x') \in \func(\rel') = \rel'$.
The neighborhood $\hyp|_{\func(\x')}$ of $\func(\x')$ in $\hyp$ 
comprises of the set of hyperedges $\setof{\rel \in \R(\hyp)}{\func(\x') \in \rel}$.
Therefore, 
$\rel' \in \R(\hyp'|_{\x'}) \iff \func(\rel') =\rel' \in \R(\hyp|_{\func(\x')})$.
Because $\T$ is a join tree of $\hyp$,
 $\hyp|_{\func(\x')}$ is connected in $\T$,
 therefore $\hyp'|_{\x'}$ is also connected in $\T$
 and $\T$ is also a join tree of $\hyp'$.
\end{proof}

Our goal will be to modify an input hypergraph $\hyp$ such that its line graph admits a monotonic weight join tree. The lemma above ensures that if this is done in such a way that a strong homomorphism between $\hyp$ and its modification is preserved then then the set of join trees of $\hyp$ is preserved.

\begin{definition}[Vertex Duplication]
Given a hypergraph $\hyp$ and a vertex $\x \in \X(\hyp)$, 
 {\em duplicating} $\x$ produces a new hypergraph $\hyp'$
 that differs from $\hyp$ by only one vertex $\{\x'\} = \X(\hyp') \setminus \X(\hyp)$
 which is added to each $\rel \in \hyp|_{\x}$.
\end{definition}

We will refer to such a hypergraph $\hyp'$
as a \emph{vertex duplication} of $\hyp$ or simply a \emph{duplicated} hypergraph of
$\hyp$. 
We generalize duplication to a set of vertices in the natural way.
 Importantly, vertex duplication preserves the set of join trees of the
original hypergraph.

\begin{lemma}\label{lemma:duplicate}
    Let $\hyp'$ be a duplicated hypergraph of $\hyp$. Then,
    \begin{enumerate}
        \item $\hyp' \twoheadrightarrow \hyp$ and $\hyp' \twoheadleftarrow \hyp$
        \item $\JTof{\hyp'} =  \JTof{\hyp}$
        \item $\lin(\hyp') = \lin(\hyp)$
    \end{enumerate}
\end{lemma}

\begin{proof}

    \textbf{1.} We only need to show that 
    $\hyp' \twoheadrightarrow \hyp$ ($\hyp' \twoheadleftarrow \hyp$) 
    holds for $\hyp'$ derived from $\hyp$ by a single vertex
    duplication. 
    $\hyp' \twoheadrightarrow \hyp$ ($\hyp' \twoheadleftarrow \hyp$) 
    will then hold for any duplicate $\hyp'$ derived by a series of
    vertex duplications since strong homomorphisms are closed under composition. 
    
    Suppose $\hyp'$ is derived from $\hyp$ by 
    duplicating $\x \in \X(\hyp)$ to $\x'$.

    Consider a pair of functions $(\func_{\X}: \X(\hyp') \mapsto \X(\hyp),
    \R(\hyp)_{\id})$ where $\func_{\X}$ maps $\x'$ to $\x$.
    Otherwise, it is an identity map $\X(\hyp)_{\id}$.
    This function pair constitutes a
    strong hypergraph homomorphism $\hyp' \twoheadrightarrow \hyp$. 
    Likewise, the pair of functions $(\X(\hyp)_{\id},
    \R(\hyp)_{\id})$  constitutes a homomorphism $\hyp
    \twoheadrightarrow \hyp'$.

    \noindent\textbf{2.}
    It follows from
    \cref{lemma:any_tree_in_H_prime} and the facts above. 

    \noindent\textbf{3.} 
    By the strong homomorphisms, the hyperedge sets have the same cardinality
     $|\R(\hyp)| = |\R(\hyp')|$. 
    So do the vertex sets of the line graphs,
    $\R(\lin(\hyp)) = \R(\hyp)$ and $\R(\lin(\hyp')) = \R(\hyp')$.
    By construction, the vertex duplication adds $\x'$ to each hyperedge
    containing $\x$, therefore does not create any additional edges in
    the line graph, $\E(\lin(\hyp')) = \E(\lin(\hyp))$.
\end{proof}

\begin{figure}[tbp]
  \centering
  \input{algos/algo_buildDG.tex}
\end{figure}

Vertex duplication allows us to modify a hypergraph while preserving its set of
join trees.
In particular, duplicating vertices can selectively increase the weights of
chosen hyperedge intersections without affecting the line graph or the join
tree structure (cf.~\cref{lemma:duplicate}).

We use this operation to enforce monotonicity of edge weights along a given join
tree.
Given an $\alpha$-acyclic hypergraph $\hyp$ together with a rooted MCS tree
$\T$, we can construct an equivalent hypergraph $\hyp^*$ by duplicating vertices
along the edges of $\T$.
The construction proceeds in a breadth-first manner from the root of $\T$ and
duplicates vertices only when necessary to ensure that each tree edge is
strictly heavier than its parent.

Obtaining an equivalent graph in this manner always succeeds.

\par\medskip
\noindent{\bf \cref{thm:MWJT}}
{\it 
 Given an $\alpha$-acyclic hypergraph $\hyp = (\X, \R, \vars)$, 
    there exists an {\em equivalent hypergraph} $\hyp^* = (\X^*, \R, \vars^*)$ over the 
    same hyperedge set $\R$ that
    \begin{itemize}
        \item admits a monotonic weight join tree $\MWJT \in \JTof{\hyp^*}$, and
        \item $\JTof{\hyp^*} = \JTof{\hyp}$.
    \end{itemize}
}

\begin{proof}
    To prove existence, we construct hypergraph $\hyp^*$ by
    vertex duplication (\cref{algo:buildDG}).
    \cref{lemma:duplicate} will guarantee that $\JTof{\hyp^*} = \JTof{\hyp}$.
    Therefore, we only need to show the existence of monotonic weight join tree $\MWJT \in \JTof{\hyp^*}$.
    \cref{algo:buildDG} takes as input hypergraph $\hyp$ along with an MCS tree $\T$ of $\hyp$. 
    \revB{
    The algorithm constructs $\hyp^*$ by vertex duplication in a
    breadth-first manner from the root of $\T$ to make it a monotonic weight
    join tree in $\hyp^*$.}

    \revB{\cref{algo:buildDG} starts with enqueuing each edge of $\T$ incident to the root.
    Once a tree edge $\e \in \E(\T_\rel)$ is dequeued,
    we check for the number of duplications needed to make it
    heavier than its parent edge $\parent(\e)$, namely
    $\Delta = \w(\parent(\e)) - \w(\e) + 1$. The duplication is 
    performed only if $\Delta > 0$.
    \cref{lemma:MCS_rip} guarantees that
    $\vars(\e) \setminus \vars(\parent(\e)) \neq \emptyset$.
    Therefore, we can always duplicate a vertex 
    $\x \in \vars(\e) \setminus \vars(\parent(\e))$
    for $\Delta$ times such that $\w(\e) > \w(\parent(\e))$.}

    We prove that this produces a monotonic weight join tree by induction on a breadth-first ordering (with respect to root $\rel$) of the edges in $\T$. 

    {\bf Base case:} The base case is an edge $\e$ incident to the root. Since $\e$ has no parent $\w(e) > \w(\parent(\e))$ holds vacuously.

    {\bf Ind. step:} Let $(\e_1, \ldots, \e_m)$ be a breadth-first ordering of the edges in $\T$ and assume $\w(e_i) > \w(\parent(\e_i))$ for all $i < k-1$. 
    
    Let $\e_k = (\rel_u, \rel_d)$ where $\depth(\T, \rel_u) < \depth(\T, \rel_d)$ and let $\x \in \vars(\e_k) \setminus \vars(\parent(\e_k))$. 
    Vertex $\x$ only occurs in the subtree rooted at $\rel_u$. Otherwise, it violates
    the running intersection property.
    Therefore, the duplication of $\x$ does not affect
    the weight of any edge $\e_i$, $i \in [k-1]$ outside of subtree rooted at $\rel_u$ and therefore $\w(e_i) > \w(\parent(\e_i))$ holds under any number of duplications of $\x$. For any $\e_i$ lying within the subtree rooted at $\rel_u$, $\e_i$ is necessarily a sibling of $\e_k$ since edges of $\T$ are breadth-first ordered. Hence $\parent(\e_i)$ lies outside this subtree and any number of vertex duplications of $\x$ cannot decrease the weight of $\e_i$ with respect to its parent. Therefore, $\w(e_i) > \w(\parent(\e_i))$ for all $i \in [k]$ after performing $\Delta$ many vertex duplications of $\x$. 
    
    Consequently, tree $\T$ has monotonic weights when~\cref{algo:buildDG} terminates.
\end{proof}

The above theorem guarantees that for any $\alpha$-acyclic hypergraph we are able to obtain an equivalent hypergraph $\hyp^*$.
Next we prove \cref{thm:algo_buildEG_on_H*} showing that \cref{algo:buildEG} produces the equivalent graph of line graph $\lin(\hyp^*)$.
\par\medskip
\noindent{\bf \cref{thm:algo_buildEG_on_H*}}
{\it 
Given a hypergraph $\hyp^* = (\X^*,\R,\vars^*)$ with
    a monotonic weight join tree $\MWJT$, \cref{algo:buildEG}
    produces an Equivalent Graph $\EG(\lin(\hyp^*))$.
}

\begin{proof}
    \revB{
    We denote the line graph of $\hyp^*$ as $\lin^*$.
    The proof is by contradiction. We assume that the resulting graph 
    $\G$ is reached from a series of valid sliding transformations with respect to $\MWJT$ but is not an equivalent graph.
    By \cref{thm:eg_max_under_sliding}, an equivalent graph $\EG(\hyp^*)$
    exists such that $\G \slide{\MWJT} \EG(\hyp^*)$.
    By \cref{def:sliding_transformation}, an edge $\e^* \in \E(\G)$ can be
    further slid towards the root of $\MWJT$. This edge has to be a
    non-tree edge $\e^* \in \E(\lin^*) \setminus \E(\MWJT)$, otherwise $\MWJT$ is not monotonic weight, violating
    \cref{def:mono_weight_tree}. Thus there must be a non-tree edge that can be slid. \cref{algo:buildEG} iterates over all non-tree edges. If an edge $\e$ has the same weight as exactly one (tree) edge $\e'$ in $\LCAE(\e)$ the algorithm slides $\e$ parallel to this edge. If edge $\e$ could slide further in $\G$, then so could tree edge $\e'$, again contradicting monotonicity of $\MWJT$. If instead $|\LCAE(\e)| = 2$ and $\e$ has the same weight as both edges $\e', \e'' \in \LCAE(\e)$ then the algorithm slides $\e$ such that it becomes incident to each of the highest depth endpoints of $\e'$ and $\e''$. By the definition of a sliding transformation $\e$ can therefore not slide.} 
    
    \revB{Of note is the case where a non-tree edge $\e$ weighs less than all edges in $\LCAE(\e)$. The algorithm then slides $\e$ to a self loop at $l = \LCA(\e)$. Such an edge may still be able to slide if edge $\{l, \parent(l)\}$ weighs more than edge $\e$. In this sense graph $\G$ returned by \cref{algo:buildEG} may not be maximal under partial order $\slide{\MWJT}$. Therefore it remains to show  that these self loops do not prevent $\G$ from being an equivalent graph. To see this, notice that no self loop can take part in a spanning tree. Therefore any spanning tree of graph obtained upon sliding self loops of $\G$ to a fixpoint ($\G^\equiv$) is also a spanning tree of $\G$. Hence, $\G$ and $\G^\equiv$ have the same set of spanning trees and $\G$ is an equivalent graph.  }
\end{proof}

The next result follows from the construction of the 
equivalent hypergraph. 
\begin{lemma}\label{lemma:MWJT-variables}
    Let $\hyp^* = (\X^*,\R,\vars^*)$ be an equivalent hypergraph derived from
    $\hyp = (\X,\R,\vars)$ by vertex duplication.
    For any non-tree edge $\e$ of an MCS tree $\T$ of $\hyp$,
    \begin{itemize}
        \item If $\vars(\e) = \vars(\e')$ for some $\e' \in \LCAE(\e)$, then $\vars^*(\e) = \vars^*(\e')$
        \item If $\vars(\e) \subset \vars(\e')$ for some $\e' \in \LCAE(\e)$, then $\vars^*(\e) \subset \vars^*(\e')$
    \end{itemize}
\end{lemma}
\begin{proof}
    The duplication of a vertex $\x \in \vars(\e)$ adds $\x'$ to both $\vars(\e)$ and $\vars(\e')$.
    If $\vars(\e) = \vars(\e')$,
    $
    \vars^*(\e) = \vars(\e) \cup \{\x'\} = \vars(\e') \cup \{\x'\} = \vars^*(\e')
    $.
    A similar argument applies if $\vars(\e) \subset \vars(\e')$.
\end{proof}

As \cref{section:enum} shows,
running \cref{algo:buildEG} on the MCS tree $\T(\hyp_6)$
in \cref{subfig:hyp_lin_jt_tmcs} produces $\EG(\hyp_6)$ identical to
 $\EG(\hyp^*_6)$ in \cref{subfig:buildEG_eg}. \cref{algo:buildEG} does this in time linear in the size of the line graph. We give the formal proof below.
\par\medskip
\noindent{\bf \cref{thm:algo_buildEG_on_H}}
{\it 
Given the line graph $\lin$ and an MCS tree $\T_\rel$ of an
    $\alpha$-acyclic hypergraph $\hyp$,
    \cref{algo:buildEG} returns an equivalent graph $\EG(\lin)$
    in time $\bigO(|\lin|)$.
}
\begin{proof}
 \revB{
 Let $\hyp^* = (\X^*,\R,\vars^*)$ be an equivalent hypergraph derived from
    $\hyp = (\X,\R,\vars)$ by vertex duplication.
    One can readily check that for any non-tree edge $\e$ of $\T$, 
    the conditions of each if statement in \cref{algo:buildEG} 
    hold under weight function $\w$ of $\lin(\hyp)$ if and only if they hold under weight function $\w^*$ of $\lin(\hyp^*)$ by 
    \cref{lemma:MWJT-variables}. Therefore \cref{algo:buildEG} returns an equivalent graph $\EG(\lin)$. }

    \revB{
    The construction of each data structure $\depth(\cdot)$, $\LA(\cdot)$ and $\LCA(\cdot)$ can be done in $\bigO(m + n)$~\cite{bender2000lca, bender2004level}. Execution of the rest of the algorithm also takes time $\bigO(m + n)$, 
    as it requires examining each non-tree edge to
    slide it directly to its LCA edges in the MCS tree.
    For each non-tree edge, the lookup of at most two LCA edges takes time $\bigO(1)$,
    as does the comparison of weights and the direct sliding transformation. The overall runtime is thus $\bigO(|\lin|)$.}
\end{proof}

Given \cref{algo:buildEG}, we now have a systematic way of enumerating join trees of an $\alpha$-acyclic hypergraph.

Starting from the line graph $\lin(\hyp)$ of an $\alpha$-acyclic hypergraph
$\hyp$, we first reconstruct $\hyp$ and apply \cref{algo:MCS} to obtain an MCS
tree $\T$.
Using $\lin(\hyp)$ and $\T$ as input, \cref{algo:buildEG} constructs an
equivalent graph of $\lin(\hyp)$.
Finally, a standard spanning tree enumeration algorithm can be applied to this
equivalent graph to enumerate all join trees of $\hyp$.
The overall running time of this process is
$\bigO(|\lin| + |\JTof{\hyp}|)$.

\par\medskip
\noindent{\bf \cref{thm:alpha_JTE_enum}}
{\it 
Given the line graph $\lin$ of an $\alpha$-acyclic hypergraph $\hyp$,
        the join trees of $\hyp$ can be enumerated in time $\bigO(|\lin| + |\JTof{\hyp}|)$.
}

\begin{proof}
\revB{
From $\lin(\hyp)$ we first retrieve query hypergraph $\hyp$. This is done in time linear in the size of the line graph. 
\cref{algo:MCS} constructs an MCS tree $\T$ in 
$\bigO(|\hyp|)$.  
With the inputs $\lin(\hyp)$ and $\T$, we can run \cref{algo:buildEG} to
construct $\EG(\hyp)$ in $\bigO(|\lin|)$ by \cref{thm:algo_buildEG_on_H}. Finally, letting $n$ be the number of vertices and $m$ the number of edges in 
$\lin(\hyp)$, the spanning 
tree enumeration algorithm by Kapoor and Ramesh~\cite{Kapoor1995} on $\EG(\hyp)$ enumerates all join trees $|\JTof{\hyp}|$ in 
$\bigO(m + n + |\JTof{\hyp}|) 
= \bigO(|\lin(\hyp)| + |\JTof{\hyp}|)$.}



\end{proof}

\subsection{\texorpdfstring{Missing Proofs in \cref{subsection:JTGammaAcyc}}{Missing Proofs in Section Enumeration}}\label{app:proofs_section4.3}
 We prove \cref{thm:gamma_JTE_enum} below.

 \par\medskip
\noindent{\bf \cref{thm:gamma_JTE_enum}}
{\it 
The join trees of a $\gamma$-acyclic $\hyp$
     can be enumerated in time
    $\bigO(|\hyp| + |\JTof{\hyp}|)$.
}

\begin{proof}
By modifying the \cref{algo:MCS} to \cref{algo:MCS4gamma}, where changes
are highlighted in red, 
we can construct
a weighted MCS tree $\T$ and an unweighted line graph $\lin(\hyp)$
 in $\bigO(|\hyp| + |\E(\lin(\hyp))|)$.
We simplify \cref{algo:buildEG} to \cref{algo:buildEG4gamma} where the changes
are highlighted in red.
As $\gamma$-acyclicity guarantees that every non-tree edge $(\rel_i, \rel_j)$ is
an MST edge, we check its LCA $l$ and the weights of LCA edges $\w_1, \w_2$
 to determine how to slide it.
If the shallower endpoint is the LCA of the deeper $\rel_i = l$, we slide
it to be parallel with the only LCA edge as in Line~\ref{algoLine:gamma_case1}.
Otherwise, if the LCA edges have different weights, 
we slide $\{\rel_i, \rel_j\}$ to be parallel to the heavier one as in 
Line~\ref{algoLine:gamma_case2} and Line~\ref{algoLine:gamma_case3}.
Finally if they have the same weight, we slide $\{\rel_i, \rel_j\}$ 
to form a triangle with them as in Line~\ref{algoLine:gamma_case4}.
Those four cases are exhaustive.

The number of join trees is lower bounded by 
the number of non-tree edges in its line graph, 
namely $|\JTof{\hyp}| = \Omega(|\E(\lin(\hyp))| - |\E(\T)|)$.
The total time complexity of enumerating all join trees is 
$\bigO(|\hyp| + |\E(\lin(\hyp))| + |\JTof{\hyp}|) = \bigO(|\hyp| + |\JTof{\hyp}|)$.
\end{proof}

\section{\texorpdfstring{Missing Proofs in \cref{section:cjt}}{Missing Proofs in Section Canonical Join Tree}}
\label{appendix:missing_proofs_cjt}

\noindent{\bf \cref{prop:berge_is_no_composite_alpha}}
{\it 
An $\alpha$-acyclic hypergraph is Berge-acyclic if and only if it is linear.
}

\begin{proof}
    Fagin~\cite{fagin1983degrees} showed that every Berge-acyclic hypergraph is linear.
    We prove the other direction by contradiction,
    assuming that $\hyp$ is $\alpha$-acyclic and linear. 
    We also assume that there is a Berge cycle 
    $(\rel_0, \x_0, \dots, \rel_{k-1}, \x_{k-1})$ of length $k \geq 2$. 
    Without loss of generality, we suppose $\rel_1$ is the first 
    hyperedge on the cycle to be removed by GYO reduction.
    Then $\rel_1$'s parent $\rel_p$ must satisfy 
    $\rel_p \cap \rel_1 
    \supseteq (\rel_1 \cap \bigcup_{i \in [k] \wedge i \neq 1} \rel_i)
    = \{\x_0, \x_1 \}$,
    contradicting $\hyp$'s linearity.
\end{proof}

\noindent{\bf \cref{lemma:berge_block_graph}}
{\it 
The line graph $\lin$ of a Berge-acyclic hypergraph $\hyp$ is a block graph.
}

\begin{proof}
A Berge-acyclic hypergraph is $\gamma$-acyclic~\cite{fagin1983degrees}.
Zhu proves that the line graph of a $\gamma$-acyclic hypergraph is 
chordal~\cite{zhu1984line}.
We now prove that the line graph is diamond-free by showing that
every diamond must be part of a $K_4$.
Assuming a diamond $(\rel_1, \rel_2, \rel_3, \rel_4)$ in $\lin$ without
the edge $(\rel_1, \rel_3)$,
it can be treated as two triangles $(\rel_1, \rel_2, \rel_4)$ 
and $(\rel_2, \rel_4, \rel_3)$ sharing an edge $(\rel_2, \rel_4)$.
We now consider the possible edge labels, $\rel_i \cap \rel_j$, in each triangle:
\begin{description}
    \item[Case 1:] If all edges are labeled with distinct variables, 
    each triangle forms a Berge-cycle, contradicting the Berge-acyclicity.
    \item[Case 2:] If two edges are labeled with the same variable, 
    all three vertices share that variable. The third edge must be 
    labeled with the same variable.
\end{description}
Since the two triangles share the edge $(\rel_2, \rel_4)$,
 all their edges are labeled with the same variable $\x$,
namely $\x \in \rel_1 \cap \rel_2 \cap \rel_3 \cap \rel_4$.
There has to be an edge $(\rel_1, \rel_3) \in \E(\lin)$,
 contradicting the assumption that $(\rel_1, \rel_2, \rel_3, \rel_4)$ is
 a diamond.
\end{proof}

\begin{lemma}\label{lemma:distinct_variables}
Let $\T_{\rel}$ be the canonical join tree rooted at $\rel$
 of a Berge-acyclic hypergraph.
Along its root-to-leaf path $\rel, \rel_1 , \ldots, \rel_{k}$,
 each pair of adjacent vertices shares a distinct variable.
\end{lemma}

\begin{proof}
    Let $\rel_0 = \rel$,
\cref{prop:berge_is_no_composite_alpha} requires that 
$|\rel_{i-1} \cap \rel_i| = 1$ for all $i\in [k]$.
By~\cref{thm:canonical_join_tree}, the path $\rel_0, \ldots, \rel_k$ 
 is the shortest path from $\rel_0$ to $\rel_k$.
 The variable shared by two adjacent vertices $\rel_{i-1}, \rel_i$
 is the edge label $\rel_{i-1} \cap \rel_i$.
 The edge labels along the path are all distinct.
 Otherwise, two consecutive edges with the same label allow for
 shortening the path by removing either,
 contradicting the shortest path property.
 Two non-consecutive edges with the same label violate the 
 running intersection property.
\end{proof}

\noindent{\bf \cref{thm:berge_cjt}}
{\it 
Given a Berge-acyclic hypergraph $\hyp$,
    \cref{algo:MCS} from $\rel \in \R(\hyp)$ constructs the 
    canonical tree $\T_{\rel}(\hyp)$.
}

\begin{proof}
Denoting the edges labeled by \cref{algo:MCS} as
$\rel = \rel_1, \dots, \rel_{n = |\R(\hyp)|}$, 
when \cref{algo:MCS} has labeled the first $i \in [n]$ hyperedges, 
we consider a graph $\G_i$ whose vertices are the labeled hyperedges 
$\{\rel_1, \ldots, \rel_i\}$
and edges are the parent-child relationships 
$\{\{\rel_2, \parent(\rel_2)\}, \ldots, \{\rel_i, \parent(\rel_i)\}\}$.

To prove the claim by induction, we show a loop invariant that
 $\G_i$ is a subtree of $\T_\rel$ for all $i \in [n]$.
An empty graph is trivially a subtree of any graph.
When the root $\rel$ is labeled as $\rel_1$, 
this graph of a single vertex is also trivially a subtree of $\T_\rel$.
Assuming $\G_k$ is a subtree of $\T_\rel$,
 \cref{algo:MCS} proceeds to label $\rel_{k + 1}$ and assigns its parent
 $\parent(\rel_{k + 1}) = \rel_{p \leq k}$.
Suppose for the sake of contradiction that
 $\rel_{k+1}$ has a different parent $\rel_q$
 in $\T_\rel$.
Let us consider the possible relationships between
 $\rel_p$ and $\rel_q$.
First, $\rel_q$ cannot be a descendant of $\rel_p$,
 otherwise there would be a shorter path from $\rel_{k+1}$ to the root $\rel$
 going through $\rel_p$ instead of $\rel_q$,
 violating~\cref{thm:canonical_join_tree}.
Suppose $\rel_q$ is an ancestor of $\rel_p$.
By the running intersection property,
 $\rel_q \cap \rel_{k+1}$
 must contain $\rel_p \cap \rel_{k+1}$,
 but that would imply every vertex in $\rel_p \cap \rel_{k+1}$
 is already marked by $\rel_q$ and the algorithm would not have
 assigned $\rel_p$ as the parent of $\rel_{k+1}$.
Therefore, $\rel_p$ and $\rel_q$ are not descendants
 of each other.
Denote their lowest common ancestor in $\T_\rel$ as 
$\rel_a = \LCA(\rel_p, \rel_q)$.
Let $\LCAE(\rel_p, \rel_q) = \{\e_l, \e_r\}$ as shown by the solid lines
in \cref{fig:mcs_to_cjt_cases},
$\spath_p$ be the path from $\rel_a$ to $\rel_p$ excluding $\e_l$ and
$\spath_q$ be the path from $\rel_a$ to $\rel_q$ excluding $\e_r$ as
shown by the dashed lines in \cref{fig:mcs_to_cjt_cases}.

\begin{figure}[tb]
    \centering
    \includegraphics[height=4cm,keepaspectratio]{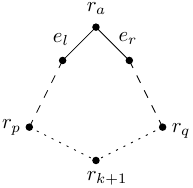}
    \caption{%
        Possible cases for how the relation $\rel_{k+1}$ attaches to the partial join tree $\G_k$
        during the construction of the canonical join tree by Maximum Cardinality Search.
    }
    \label{fig:mcs_to_cjt_cases}
\end{figure}

By~\cref{lemma:distinct_variables}, the variables on $\spath_p$ and $\spath_q$ are all distinct.
Consider the following possible cases where 
$\vars^l = \bigcap_{\rel \in \nodes(\e_l)}{\rel}$ 
and $\vars^r = \bigcap_{\rel \in \nodes(\e_r)}{\rel}$:
\begin{description}
    \item[Case 1:] If $\rel_{k+1} \cap \rel_p = \rel_{k+1} \cap \rel_q$
    and $\vars^l = \vars^r$,
    then $\spath_p, \spath_q$
    form a Berge-cycle.
    \item[Case 2:] If $\rel_{k+1} \cap \rel_p = \rel_{k+1} \cap \rel_q$
    and $\vars^l \neq \vars^r$,
    then $\spath_p, \rel_a, \spath_q$ form a Berge-cycle.
    \item[Case 3:] If $\rel_{k+1} \cap \rel_p \neq \rel_{k+1} \cap \rel_q$
    and $\vars^l = \vars^r$,
    then $\spath_p, \spath_q, \rel_{k+1}$ form a Berge-cycle.
    \item[Case 4:] If $\rel_{k+1} \cap \rel_p \neq \rel_{k+1} \cap \rel_q$
    and $\vars^l \neq \vars^r$,
    then $\spath_p, \rel_a, \spath_q, \rel_{k+1}$ form a Berge-cycle.
    
\end{description}
%
%
All possible cases above contradict Berge-acyclicity,
 therefore, \cref{algo:MCS} must correctly assign $\rel_p$ as the
 parent of $\rel_{k+1}$.
\end{proof}

\section{\texorpdfstring{Missing Algorithm and Proofs in \cref{section:frombinplan}}{Missing Proofs in Section Binary Plan conversion}}
\label{appendix:missing_proofs_frombinary}

\noindent{\bf \cref{thm:bin2tree}}
{\it 
A query is $\gamma$-acyclic
 if and only if every connected left-deep linear join plan
 for the query is the reverse of a GYO-reduction order.
}

We first prove ``only if'' direction of \cref{thm:bin2tree} by contradiction,
assuming a query plan with
 a $\gamma$-cycle that is not the reverse of a 
 GYO-reduction order.
%
\begin{proposition}\label{prop:orphan}
If a left-deep linear plan is not the reverse of a GYO-reduction order,
 then there exists a relation $\rel_i$ that has no parent among $\{\rel_{j<i}\}$,
 formally $\neg \exists \rel_{j<i} : (\rel_i \cap \bigcup \rel_{k < i}) \subseteq \rel_j$.
We call $\rel_i$ an {\em orphan}, and denote it as $\orphan{\rel_i}$.
\end{proposition}
We introduce an ordering on relations to compare them in a plan.
\begin{definition}
Given a relation $\rel$, we write $\rel_1 \geq_{\rel} \rel_2$ 
if $\rel_1 \cap \rel \supseteq \rel_2 \cap \rel$.
\end{definition}
We first observe that $\geq_{\rel}$ is a preorder (reflexive and transitive), 
but not a partial order (antisymmetry may fail).
Moreover, if there is a greatest element $\rel_p$
 with respect to $\geq_{\rel}$,
 i.e., $\forall r' : \rel_p \geq_{\rel} r'$,
 then $\rel_p$ is a parent of $\rel$.
Those observations lead to the following.
\begin{lemma}\label{lem:maximal}
Let $\hat{\rel}$ be a {\em maximal} relation with respect to $\geq_{\orphan{\rel}}$
 for an orphan $\orphan{\rel}$,
 meaning $\neg \exists r : r >_{\orphan{\rel}} \hat{\rel}$,
 then there is a relation $\bar{\rel}$ incomparable with $\hat{\rel}$,
 i.e., $\hat{\rel} \not\geq_{\orphan{\rel}} \bar{\rel}$ and $\bar{\rel} \not\geq_{\orphan{\rel}} \hat{\rel}$.
\end{lemma}
\begin{proof}
If every relation is comparable with $\hat{\rel}$,
 the maximal relation $\hat{\rel}$ is the greatest element with respect to $\geq_{\orphan{\rel}}$,
 thus a parent of $\orphan{\rel}$.
 It contradicts the fact that $\orphan{\rel}$ is an orphan. 
\end{proof}
Given a pair of incomparable relations $\hat{\rel}$ and $\bar{\rel}$ with 
respect to $\orphan{\rel}$, we can find at least one variable in each relation
 that does not appear in the other relation but appears in $\orphan{\rel}$,
 namely $\exists \bar{\x} \in \bar{\rel} \cap \orphan{\rel} : \bar{\x} \notin \hat{\rel}$
 and $\exists \hat{\x} \in \hat{\rel} \cap \orphan{\rel} : \hat{\x} \notin \bar{\rel}$.
 We say that $\hat{\x}$ is a {\em dangling variable} of $\hat{\rel}$ and 
 $\bar{\x}$ is a {\em dangling variable} of $\bar{\rel}$.
We are now ready to prove \cref{thm:bin2tree}.
Our strategy is to derive a $\gamma$-cycle from
 any orphan relation, leading to a contradiction.

\begin{proof}[Proof of \cref{thm:bin2tree}]
To prove the ``only if'' direction by contradiction,
 we assume that a connected left-deep linear plan for a $\gamma$-acyclic query
 is not the reverse of a GYO-reduction order.
By~\cref{prop:orphan}, there is an orphan $\orphan{\rel}$ in the plan.
Assuming that $\orphan{\rel}$ is the $t$-th relation in the plan,
we now consider the prefix of the plan up to $\orphan{\rel}$, namely
$\prefix{t} = \rel_1, \ldots, \rel_{t-1}, \orphan{\rel}$.
Let $\hat{\rel} \in \prefix{t-1}$ be a maximal relation with respect to 
$\geq_{\orphan{\rel}}$,
 \cref{lem:maximal} guarantees an $\bar{\rel} \in \prefix{t-1}$ 
 incomparable with $\hat{\rel}$. 
 There are at least two dangling variables
 $\bar{\x} \in \bar{\rel} \cap \orphan{\rel} \setminus \hat{\rel}$ and
 $\hat{\x} \in \hat{\rel} \cap \orphan{\rel} \setminus \bar{\rel}$.

Because the shorter prefix $\prefix{t-1} = \rel_1, \ldots, \rel_{t-1}$ 
is connected, there is a path in $\prefix{t-1}$ between $\hat{\rel}$ and 
$\bar{\rel}$.
The incomparability guarantees that
$\bar{\rel} \cap \orphan{\rel} \not \subseteq \hat{\rel} \cap \orphan{\rel}$.
We can choose a shortest path in $\prefix{t-1}$ that connects a variable in 
$\hat{\rel} \cap \orphan{\rel}$ 
to a variable in $\orphan{\rel} \setminus \hat{\rel}$.
We now consider two exhaustive cases where each yields 
a $\gamma$-cycle:
\begin{description}
  \item[Case 1:] If the shortest path consists of a single relation $\tilde{\rel}$,
  $\tilde{\rel} \not >_{\orphan{\rel}} \hat{\rel}$ (because $\hat{\rel}$ is maximal)
  and $\tilde{\rel} \not \leq_{\orphan{\rel}} \hat{\rel}$ 
  (because $\tilde{\rel}$ contains some variable in $\orphan{\rel} \setminus \hat{\rel}$).
  Therefore, $\tilde{\rel}$ is incomparable with $\hat{\rel}$.
  Since $\tilde{\rel} \cap \hat{\rel} \neq \emptyset$,
  $\hat{\rel}, \hat{\x}, \orphan{\rel}, \tilde{\x}, \tilde{\rel}, \breve{\x}$ 
  form a $\gamma$-cycle
  where $\hat{\x} \in \hat{\rel} \cap \orphan{\rel}$,
  $\tilde{\x} \in \tilde{\rel} \cap \orphan{\rel}$, and $\breve{\x} \in \tilde{\rel} \cap \hat{\rel}$.
  \item[Case 2:] If the shortest path consists of at least two relations,
  let $\rel_0 = \orphan{\rel}$,
  and the path be $\rel_1, \ldots, \rel_{k-1}$.
  We now show that $\rel_0, \rel_1, \ldots, \rel_{k-1}$ form a {\em pure cycle} 
  where each $\rel_i$ is only adjacent
  to $\rel_{(i-1) \bmod k}$ and $\rel_{(i+1) \bmod k}$. 
  First, if a $\rel_{i \in [2, k-2]}$ on the path is adjacent to a 
  $\rel_{j \in [1, k-1]}$ 
  other than its neighbors $\rel_{i-1}$ or $\rel_{i+1}$,
  the path can be shortened, contradicting the shortest path assumption.
  Second, if the relation $\rel_{i \in [2, k-2]}$ is adjacent to $\rel_0$, 
  then $\rel_i$ must contain variables in 
  either $\rel_0 \cap \hat{\rel}$ or $\rel_0 \setminus \hat{\rel}$.
  We can shrink the path to $\rel_1, \ldots, \rel_i$ or
  $\rel_i, \rel_{k-1}$, which also contradicts the shortest path assumption.
  Therefore, each $\rel_{i \in [0, k-1]}$ on the cycle is only adjacent
  to $\rel_{(i-1) \bmod k}$ and $\rel_{(i+1) \bmod k}$.
  The cycle $\rel_0, \rel_1, \ldots, \rel_{k-1}$ is a pure cycle
  of length $k \geq 3$.
  A pure cycle is always a $\gamma$-cycle.
\end{description}
This concludes the proof for the ``only if'' direction.

Now we prove the ``if'' direction by contrapositive,
 namely that if a query contains a $\gamma$-cycle,
 then there exists a connected left-deep linear plan
 that is not the reverse of a GYO-reduction order.
Let the $\gamma$-cycle of length $k \geq 3$ be
$\rel_0, \ldots, \rel_i, \ldots, \rel_{k-1}$ where $\rel_{i \in [1, k-2]}$.
By definition, there is an $x_{i-1}$ appearing exclusively in
 $\rel_{i-1}$ and $\rel_i$ 
 and an $x_{i}$ appearing exclusively in $\rel_i$ and $\rel_{i+1}$.
We modify the original plan by moving $\rel_i$ after $\rel_{k-1}$
such that 
$\rel_0, \ldots, \rel_{i-1}, \rel_{i+1}, \ldots, \rel_{k-1}, \rel_i$
becomes part of the new query plan.

The prefix $\prefix{k - 1} = \rel_0, \ldots, \rel_{i-1}, \rel_{i+1}, \ldots, \rel_{k-1}$
is connected. Therefore, the new query plan is connected.
However, $\rel_i$ has no parent,
because no relation in $\prefix{k - 1}$ contains both $x_{i-1}$ and $x_{i}$. This implies the plan is not the reverse of a GYO-reduction order.
\end{proof}

\section{Notation Table}\label{appendix:additional_fig_tab_algo}

\begin{table}[t]
\centering
\small
\setlength{\tabcolsep}{5pt}
\renewcommand{\arraystretch}{1.2}
\begin{tabular}{p{0.24\linewidth} p{0.71\linewidth}}
\hline
\textbf{Variable} & \textbf{Definition}\\
\hline
$\hyp, \X, \R$ & hypergraph, vertex set, hyperedge set \\
$\vars:\R\to 2^{\X}$ & Incidence function mapping each hyperedge to its vertices.\\
$\X(\hyp), \R(\hyp)$ & Vertex set and edge set of $\hyp$.\\
$|\hyp|$ & Size of hypergraph: $|\hyp|=\sum_{\rel\in\R}|\rel|$.\\
$\hyp|_{\x}$ & Neighborhood of $\x$ in $\hyp$: set of hyperedges containing $\x$.\\
$\hyp^*, \X^*, \vars^*$ & Equivalent hypergraph, its vertices, and incidence function \\
$\func$ & A hypergraph homomorphism (pair $(\func_{\X},\func_{\R})$).\\
$\hyp_1 \rightarrow \hyp_2, \hyp_1 \twoheadrightarrow \hyp_2$ & Homomorphism from $\hyp_1$ to $\hyp_2$, strong homomorphism from $\hyp_1$ to $\hyp_2$.\\
\hline
$\G, \R, \E$ & (Multi-)graph, nodes (corresponding to hyperedges), edges\\
$\nodes:\E\to 2^{\X}, \w: \E\to \mathbb{N}$ & Incidence function mapping each edge to its endpoints, weight function\\
$\R(\G), \E(\G)$ & Vertex set and edge set of $\G$.\\
$\EG, \nodesEG, \T^\equiv$ & Equivalent graph (Def.~\ref{def:equivalent_graph}), its incident function, and a spanning tree\\
$\G \slide{\T} \G'$ & $\G$ can be slid into $\G'$ (Def.~\ref{def:sliding_transformation}).\\
\hline
$\lin(\hyp), \lin, \vars(\e), \w(\e)$ & Line graph of $\hyp$, edge label $\vars(\{\rel_1, \rel_2\}) = \rel_1\cap\rel_2$, weight $\w(\e) = |\vars(\e)|$.\\
$|\lin|$ & Size of line graph: $|\lin|=\sum_{\e\in\E(\lin)}\w(\e)$.\\
$\G|_{\x}$ & Subgraph of $\G \subseteq \lin(\hyp)$  induced by $\hyp|_{\x}$.\\
\hline
$\T^G, \T^M$  & Join tree generated by GYO and MCS, resp.\\
$\T, \T_{\rel}$ & Unrooted and rooted (at $\rel$) join tree.\\
$\MWJT$, $\MWJT_{\rel}$ & A monotonic weight join tree (Def.~\ref{def:mono_weight_tree}), one rooted at $\rel$.\\
$\JTof{\lin(\hyp)}$, $\JTof{\hyp}$ & Set of (unrooted) join trees of $\lin(\hyp)$ / $\hyp$.\\
$\depth(\T_{\rel},\rel_i)$ & Depth of node $\rel_i$ in rooted tree $\T_{\rel}$.\\
$\rel_p$ & The parent of $\rel_i$ in a GYO reduction order.\\
$\children(\rel), \parent(\e), \siblings(\e)$ & Children nodes of $\rel$, parent edge of $\e$, sibling edges of $\e$\\

$\LCA(\rel_i,\rel_j), \LA(\cdot)$ & Lowest common ancestor and level-ancestor.\\
$\LCAE(\e)$ & LCA edges of a non-tree edge $\e=(\rel_i,\rel_j)$.\\
$\Delta$ & Number of duplications needed to enforce monotonicity.\\


\hline
\end{tabular}
\caption{Notation introduced in sections \ref{section:preliminaries}--\ref{section:enum}.}
\label{tab:notation-prelims-enum}
\end{table}









\end{document}